\documentclass[journal]{IEEEtran}
\usepackage{cite}

\ifCLASSINFOpdf
   \usepackage[pdftex]{graphicx}
   \graphicspath{{figures/}}
\else
   \usepackage[dvips]{graphicx}
\fi
\usepackage{algorithmicx}
\usepackage{algpseudocode}
\ifCLASSOPTIONcompsoc
  \usepackage[caption=false,font=normalsize,labelfont=sf,textfont=sf,subrefformat=parens]{subfig}
\else
  \usepackage[caption=false,font=footnotesize,subrefformat=parens]{subfig}
\fi
\usepackage{dblfloatfix}

\ifCLASSOPTIONcaptionsoff
  \usepackage[nomarkers]{endfloat}
 \let\MYoriglatexcaption\caption
 \renewcommand{\caption}[2][\relax]{\MYoriglatexcaption[#2]{#2}}
\fi
\usepackage{url}

\usepackage{mathtools}
\usepackage{amssymb}
\usepackage[capitalize]{cleveref}
\crefname{equation}{}{}
\usepackage{longtable}
\usepackage{multirow}
\usepackage{bbm}
\usepackage{bm}
\usepackage{amsthm}
\usepackage{commath}
\usepackage{framed}

\newtheorem{proposition}{Proposition}
\theoremstyle{definition}

\DeclareMathOperator{\sinc}{sinc}
\newcounter{tempEquationCounter} 
\newcounter{thisEquationNumber}
\newenvironment{floatEq}
{\setcounter{thisEquationNumber}{\value{equation}}\addtocounter{equation}{1}
\begin{figure*}[!t]
\normalsize\setcounter{tempEquationCounter}{\value{equation}}
\setcounter{equation}{\value{thisEquationNumber}}
}
{\setcounter{equation}{\value{tempEquationCounter}}
\hrulefill\vspace*{4pt}
\end{figure*}
}

\begin{document}
\bstctlcite{IEEEexample:BSTcontrol}
%
\title{Uplink Channel Estimation and Data Transmission in Millimeter-Wave CRAN with Lens Antenna Arrays}
%
%
%

\author{Reuben~George~Stephen,~\IEEEmembership{Student Member,~IEEE,}
        and~Rui~Zhang,~\IEEEmembership{Fellow,~IEEE}
\thanks{Copyright \copyright 2018 IEEE. Personal use of this material is permitted. However, permission to use this material for any other purposes must be obtained from the IEEE by sending a request to pubs-permissions@ieee.org.}
\thanks{This paper was presented in part at the IEEE Global Communications Conference~(GLOBECOM), Singapore, Dec.\ 4--8, 2017.}
\thanks{This  work  was  supported  in  part  by  the  National  University  of Singapore under Research  Grant R-263-000-B62-112.}
\thanks{R. G. Stephen is with the National University of Singapore Graduate School for Integrative Sciences and Engineering, National University of Singapore, Singapore 117456, and also with the Department of Electrical and Computer Engineering, National University of Singapore, Singapore 117583 (e-mail: reubenstephen@u.nus.edu).}
\thanks{R. Zhang is with the Department of Electrical and Computer Engineering, National University of Singapore, Singapore 117583 (e-mail: elezhang@nus.edu.sg).}
}

\maketitle

\begin{abstract}
Millimeter-wave~(mmWave) communication and network densification hold great promise for achieving high-rate communication in next-generation wireless networks. 
Cloud radio access network~(CRAN), in which low-complexity remote radio heads~(RRHs) coordinated by a central unit~(CU) are deployed to serve users in a distributed manner, is a cost-effective solution to achieve network densification. However, when operating over a large bandwidth in the mmWave frequencies, the digital fronthaul links in a CRAN would be easily saturated by the large amount of sampled and quantized signals to be transferred between RRHs and the CU. To tackle this challenge, we propose in this paper a new architecture for mmWave-based CRAN with advanced lens antenna arrays at the RRHs. Due to the energy focusing property, lens antenna arrays are effective in exploiting the angular sparsity of mmWave channels, and thus help in substantially reducing the fronthaul rate and simplifying the signal processing at the multi-antenna RRHs and the CU, even when the channels are frequency-selective. We consider the uplink transmission in a mmWave CRAN with lens antenna arrays and propose a low-complexity quantization bit allocation scheme for multiple antennas at each RRH to meet the given fronthaul rate constraint. Further, we propose a channel estimation technique that exploits the energy focusing property of the lens array and can be implemented at the CU with low complexity. 
Finally, we compare the proposed mmWave CRAN using lens antenna arrays with a conventional CRAN using uniform planar arrays at the RRHs, and show that the proposed design achieves significant throughput gains, yet with much lower complexity.
\end{abstract}
%
\IEEEpeerreviewmaketitle
\section{Introduction}
%
%
%
%
\IEEEPARstart{N}{etwork} densification by increasing the densities of base stations~(BSs)/access points~(APs) deployed, and millimeter-wave~(mmWave) communication by exploiting large unused bandwidth in the higher frequencies of the radio spectrum, are two key strategies for achieving the orders of magnitude data rate improvement required for future wireless communication networks~\cite{andrews-etal2014what}. On one hand, cloud radio access network~(CRAN), in which distributed low-complexity remote radio heads~(RRHs) are deployed close to the users, and coordinated by a central unit~(CU) for joint processing, provides a cost-effective way of achieving network densification. Thanks to the centralized resource allocation and joint signal processing for the RRHs at the CU, CRAN achieves significant improvements in both the spectral efficiency and energy efficiency compared to the conventional cellular network~\cite{zhou-yu2014optimized,liu-zhang2015optimized,dai-yu2014sparse,shi-etal2014group,luo-etal2015downlink,stephen-zhang2017joint,liu-etal2015joint,stephen-zhang2017fronthaul}, while RRHs can be simple relay nodes that transmit and receive quantized/compressed baseband signals to/from the CU over their fronthaul links. On the other hand, with advances in radio frequency~(RF) circuits, wireless communication over mmWave bands has emerged as a promising technology to achieve high-rate communications, due to the large bandwidth available and the beamforming gain brought about by the possibility of deploying a large number of antennas at the transceivers thanks to the small wavelengths~\cite{bai-etal2014coverage,akdeniz-etal2014millimeter}. Thus, CRAN when integrated with mmWave communication, achieves the double goals of network densification and ample bandwidth at the same time for future wireless networks. However, in such a dense CRAN operating over a large mmWave bandwidth, the digital fronthaul links would be easily saturated by the large volume of sampled and quantized/compressed baseband signals that need to be transmitted between the CU and RRHs; thus, it is crucial to find cost-effective solutions to reduce the transmission rate required for each fronthaul link. 

In CRANs operating over the conventional cellular frequency bands, a considerable body of prior work~(e.g.~\cite{zhou-yu2014optimized,dai-yu2014sparse,liu-zhang2015optimized,shi-etal2014group}) has investigated various techniques for data compression at the RRHs to achieve fronthaul rate reduction. 
Further, channel estimation in networks with coordinated multi-antenna BSs has also been considered in prior work~\cite{hoydis-etal2011optimal,kang-etal2014joint}. 
In~\cite{hoydis-etal2011optimal}, random matrix theory was used to derive an approximate lower bound on the uplink ergodic achievable rate in a system with multiple single-antenna users and multi-antenna BSs, while~\cite{kang-etal2014joint} considered an estimate-compress-forward approach for the uplink of a CRAN and proposed various schemes to optimize the ergodic achievable sum rate subject to backhaul constraints. Most of the above techniques for compression and/or channel estimation are applicable for relatively small bandwidth compared to that in mmWave, and typically involve complex signal processing and cooperative signal compression across the RRHs, which are difficult to implement for mmWave systems due to practical cost and complexity considerations. A low-complexity training sequence design for CRAN was considered in~\cite{zhang-etal2017locally}, where the problem of minimizing the training length, while maintaining local orthogonality among the training sequences of the users, was considered; however, the fronthaul constraints at the RRHs were ignored. 
For frequency-selective mmWave channels, 
channel estimation 
for hybrid precoding was considered in~\cite{gao-etal2016channel,venugopal-etal2017channel}. The approach is to represent the channel taps in the angular domain corresponding to a set of quantized angles, and then use sparse signal processing techniques to estimate the relevant parameters. 

In this paper, we propose a new architecture for mmWave CRAN by leveraging the use of advanced \emph{lens antenna arrays}~\cite{zeng-etal2014electromagnetic,zeng-zhang2016millimeter,zeng-zhang2017cost,zeng-etal2016multiuser,kwon-etal2016rf,yang-etal2017channel} at the RRHs. A full-dimensional lens antenna array~\cite{zeng-zhang2017cost} consists of an electromagnetic~(EM) lens with energy focusing capability integrated with an antenna array whose elements are located on the focal surface of the lens~(see~\cref{F:mmWLLensArray}). The amplitude response of a lens array can be expressed as a ``sinc"-type function in terms of the angles of arrival/departure of plane waves incident on/transmitted from it, and the locations of the antenna elements~\cite{zeng-zhang2016millimeter,zeng-zhang2017cost}. 
Hence, by appropriately designing the locations of the antenna elements on the focal surface, the lens array is capable of focusing most of the energy from a uniform plane wave arriving in a particular direction onto a specific antenna element or subset of elements. Moreover, due to the multi-path sparsity of mmWave channels~\cite{akdeniz-etal2014millimeter}, a lens array can be used to achieve the capacity of a point-to-point mmWave channel with multiple antennas via a technique called path division multiplexing~\cite{zeng-zhang2016millimeter}, which uses simple single-carrier modulation even for transmission over wideband frequency-selective channels, and has low signal processing complexity. 
\begin{figure}[t]
\centering
\includegraphics[width=\columnwidth]{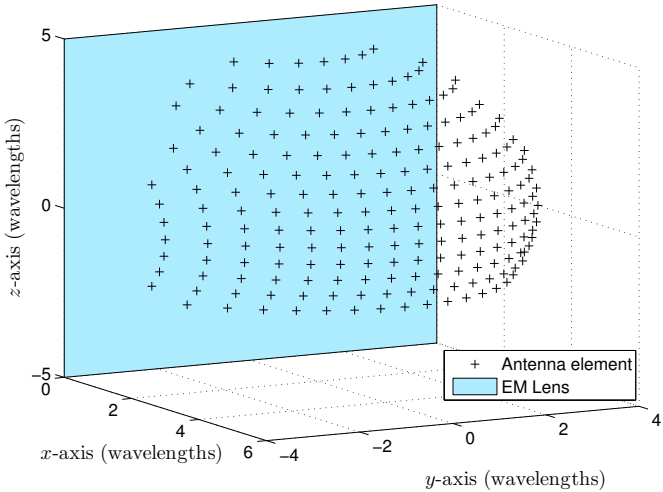}
\caption{Illustration of lens antenna array with $D_y=D_z=10,\Theta_-=\tfrac{\pi}{6}
,\Theta_+=\tfrac{\pi}{2}
,
\Phi_-=\Phi_+=\tfrac{\pi}{2}
$.}\label{F:mmWLLensArray}
\end{figure}

With lens antenna arrays, 
the angular domain\footnote{Also known as beamspace channel~\cite{brady-etal2013beamspace}.} sparsity of mmWave channels is transformed to the spatial domain, 
which then enables lower complexity signal processing for channel estimation and data transmission.  Beamspace channel estimation based on compressed sensing techniques was considered in~\cite{gao-etal2017reliable} for a mmWave system with a multi-antenna BS and multiple users. However, the channel was assumed to be frequency non-selective, unlike the general freqeuncy-selective channel in this paper. Moreover, in a CRAN, since the angles of arrival of the signals from different users are typically independent, and thus different at each RRH, the user signals are effectively separated over a small number of focusing antennas that can be selected for signal sampling and quantization. Thus, the use of lens antenna arrays can potentially achieve significant reduction in both the fronthaul rate requirement for transmitting the quantized signals to the CU and the interference among the users for the joint decoding at the CU in the uplink transmission. However, due to the finite fronthaul rate constraint at each RRH, it is crucial to design antenna selection and signal quantization schemes to maximize the achievable user rates at the CU. Since the RRHs in a CRAN are typically simple relay nodes, in this paper, we consider that simple uniform scalar quantization~(SQ) of the sampled baseband signals is performed independently over the antennas selected at each RRH. The major contributions of this paper are summarized as follows. 
\begin{itemize}
\item We introduce a new architecture for CRAN in mmWave frequencies using full-dimensional lens antenna arrays~\cite{zeng-zhang2016millimeter,zeng-zhang2017cost} at the RRHs, taking into consideration both the elevation and azimuth angles of arrival of the signals from the users in the uplink transmission. 
\item We propose a simple energy-detection based antenna selection at each RRH and a low-complexity quantization bit allocation algorithm over the selected antennas to minimize the total quantization noise power based only on the estimates of the received signal power at different antennas, subject to the fronthaul rate constraint. 
\item When the CU has perfect channel state information~(CSI), we show that the proposed system with lens antenna arrays can achieve better sum-rate performance compared to a conventional CRAN using uniform planar arrays~(UPAs) at the RRHs and orthogonal frequency division multiplexing~(OFDM) transmission by the users, when the same quantization algorithm is employed in both cases. 
\item Under imperfect CSI, we propose a \emph{reduced-size, approximate} MMSE beamforming, by exploiting the energy-focusing property of lens antenna arrays. In the proposed scheme, for each user, only the data streams where the estimated channel gains are larger than a certain threshold are chosen for beamforming, while the interference on these streams is also approximated by thresholding the channel gain estimates. 
\item With the proposed bit allocation at the RRHs, channel estimation at the CU, and data transmission with single-carrier modulation, we compare the proposed system with the conventional CRAN with UPAs via simulations, and show that when the fronthaul is constrained, the proposed system can achieve significant sum-throughput gains at much lower signal processing complexity and training overhead compared to the benchmark. 
\end{itemize}
The rest of this paper is organized as follows. In~\cref{Sec:mmWLSysMod} we present the system model for the proposed mmWave CRAN with lens arrays, and carry out an analysis of the achievable rates under perfect and imperfect CSI, with our proposed bit allocation and channel estimation schemes. \cref{Sec:mmWLUPA} gives a brief description of the benchmark CRAN system with UPAs and OFDM transmission. In~\cref{Sec:mmWLSimRes}, we compare our proposed system with the benchmark system via simulations. Finally,~\cref{Sec:mmWLSumm} concludes the paper.

{\it Notation}: In this paper, $\triangleq$ denotes equality by definition, and $\sim$ means ``distributed as". The cardinality of a finite set $\mathcal{S}$ is denoted by $|\mathcal{S}|$, while $\mathcal{A}\setminus\mathcal{B}$ denotes the elements in $\mathcal{A}$ that are not in $\mathcal{B}$. Sets of $M\times N$ real, complex, and integer matrices are denoted by $\mathbb{R}^{M\times N}$, $\mathbb{C}^{M\times N}$ and $\mathbb{Z}^{M\times N}$ respectively, while $\mathbb{R}_+$, $\mathbb{Z}_+$, 
and $\mathbb{Z}_{++}$ 
denote the set of non-negative real numbers, non-negative integers, and positive integers, respectively. 
$\delta[n]$ denotes the Kronecker delta function. 
The normalized sinc function is defined as $\sinc(x)\triangleq\tfrac{\sin(\pi x)}{\pi x}$ for $x\neq0$ and $\sinc(x)\triangleq1$ if $x=0$. 
The imaginary unit is denoted by $\jmath$ with $\jmath^2=1$. Scalars are denoted by lower-case letters, e.g. $x$, while vectors and matrices are denoted by bold-face lower-case and upper-case letters,  e.g. $\bm x$ and $\bm X$, respectively. For $x\in\mathbb{R}$, $\lceil x\rceil$ denotes the smallest integer greater than or equal to $x$, and $\lfloor x\rfloor$ denotes the largest integer less than or equal to $x$. 
For $x\in\mathbb{C}$, $|x|\geq 0$ denotes its magnitude and 
$\angle x\in[0,2\pi)$ denotes its phase in radian. 
For a vector $\bm x$, $\|\bm x\|$ denotes its Euclidean norm. A vector with all elements equal to 
$0$ is denoted by 
$\bm 0$, 
where the dimension is implied from the context. 
For vectors and matrices, $()^\mathsf{T}$ denotes transpose, and 
$()^\mathsf{H}$ denotes conjugate transpose (Hermitian). 
For matrix $\bm X$, $\mathrm{tr}(\bm X)$ denotes the sum of its diagonal elements~(trace). 
For $\bm X$ with linearly independent columns, $\bm X^\dag\triangleq(\bm X^\mathsf{H}\bm X)^{-1}\bm X^\mathsf{H}$ denotes the Moore-Penrose pseudo-inverse. $\bm X\otimes\bm Y$ denotes the Kronecker product, 
while $\bm I_N$ denotes the identity matrix of dimension $N$.  A diagonal matrix with elements $x_1,\dotsc,x_M$ on the main diagonal is denoted by $\mathrm{diag}\begin{psmallmatrix}x_1&\cdots & x_M\end{psmallmatrix}$, 
and a block diagonal matrix 
by $\mathrm{blkdiag}\begin{psmallmatrix}\bm X_1&\cdots &\bm X_M\end{psmallmatrix}$.  $\mathcal{CN}(\bm \mu,\bm\Sigma)$ denotes a circularly symmetric complex Gaussian~(CSCG) distribution centered at $\bm\mu$ with covariance $\bm\Sigma$, and $\mathcal{U}[a,b]$ denotes a uniform distribution over the interval $[a,b]$. 
\section{System Model}\label{Sec:mmWLSysMod}
Before introducing the system model, we first summarize all the notations used in this paper in~\cref{T:mmWLListSym} for ease of reference. 
\begin{table}[t]
{\caption{List of Symbols and Their Meanings}
\label{T:mmWLListSym}
\centering
\begin{tabular}{p{3cm}|p{5cm}}\hline\hline Symbol & Meaning \\ \hline\hline $\mathcal{M}=\{1,\dotsc,M\}$ & Set of RRHs\\ $\mathcal{J}=\{1,\dotsc,J\}$ & Set of sectors at each RRH\\ $\mathcal{Q}^{(j)}=\{1,\dotsc,Q\}$; respectively (resp.), $\tilde{\mathcal{Q}}^{(j)}=\{1,\dotsc,\tilde{Q}\}$& Set of antennas at sector $j$ of each RRH with lens antenna arrays (resp.\ UPAs)\\ $\mathcal{Q}=\bigcup_{j\in\mathcal{J}}\mathcal{Q}^{(j)}$ (resp.\ $\tilde{\mathcal{Q}}=\bigcup_{j\in\mathcal{J}}\tilde{\mathcal{Q}}^{(j)}$) & Set of all antennas at each RRH with lens array (resp.\ UPA)\\  $\mathcal{Q}_m$ with $|\mathcal{Q}_m|=Q_m$ (resp.\ $\tilde{\mathcal{Q}}_m$ with $|\tilde{\mathcal{Q}}_m|=\tilde{Q}_m$& Set of antennas selected at RRH $m$ with lens antenna array (resp.\ UPA)\\ $Q_\mathrm{tot}$ & Total number of selected antennas\\ $\mathcal{I}=\{1,\dotsc,Q_\mathrm{tot}\}$ & Set of all streams (selected antennas)\\$\mathcal{K}=\{1,\dotsc,K\}$ & Set of users \\ $\mathcal{L}_{m,k}=\{1,\dotsc,L_{m,k}\}$ & Set of paths from user $k$ to RRH $m$\\ $D_y\times D_z$ or $\tilde{D}_y\times \tilde{D}_z$ & Rectangular dimensions of EM lens or UPA in $y$-$z$ plane (normalized by wavelength)\\ 
$\ell^\star_{i,k}$ & Maximum gain path from user $k$ for stream $i$\\ $d_{m_i,k,\ell^\star_{i,k}}$~(resp.\ $\breve{d}_{i,k}$) & Tap delay (resp.\ estimate) of maximum gain path from user $k$ for stream $i$\\
$h_{i,k}[n]$~(resp.\ $\hat{h}_{i,k}[n]$) & Time-domain channel coefficient~(resp.\ estimate) for stream $i$ from user $k$ at time $n$\\$x_{\mathrm{d},k}[n]$ (resp.\ $x_{\mathrm{p},k}[n]$) & Time-domain data (resp.\ pilot) symbol transmitted by user $k$ at time $n$\\
$\check{y}_{\mathrm{d},i}[n]$ (resp.\ $\check{y}_{\mathrm{p},i}[n]$) & Time-domain quantized data (resp.\ pilot) signal for stream $i$ and time $n$\\ $\bar{\check{y}}_{\mathrm{d},i,k}[n]=\check{y}_{\mathrm{d},i,k}[n+d_{m_i,k,\ell^\star_{i,k}}]$~(resp.\ $\breve{\check{y}}_{\mathrm{d},i,k}[n]=\check{y}_{\mathrm{d},i,k}[n+\breve{d}_{i,k}]$) & Time-domain delay compensated quantized data signal according to delay $d_{m_i,k,\ell^\star_{i,k}}$~(resp.\ estimated delay $\breve{d}_{i,k}$) of maximum gain path from user $k$ 
\\$\bar{h}_{i,kk'}[\nu]$ (resp.\ $\breve{h}_{i,kk'}[\breve{\nu}]$) &  Sum of channel coefficients corresponding to paths $\ell'\in\mathcal{L}_{m_i,k'}$ from user $k'$ for stream $i$, which have a delay difference $\nu$~(resp.\ $\breve{\nu}$) with $d_{m_i,k,\ell^\star_{i,k}}$~(resp.\ $\breve{d}_{i,k}$)\\$\hat{\breve{h}}_{i,kk'}[\breve{\nu}]$ & Thresholded estimates of $\breve{h}_{i,kk'}[\breve{\nu}]$\\
$\tilde{h}_{i,k}[n]$~(resp.\ $\hat{\tilde{h}}_{i,k}[n]$) & Frequency-domain channel coefficient~(resp.\ estimate) for stream $i$ from user $k$ on SC $n$\\ $\tilde{x}_{\mathrm{d},k}[n]$ & Frequency-domain data symbol transmitted by user $k$ on SC $n$\\$\tilde{\check{y}}_{\mathrm{d},i}[n]$ (resp.\ $\tilde{\check{y}}_{\mathrm{p},t,i}$) & Frequency-domain quantized data (resp.\ pilot) signal for stream $i$, and SC $n$\\ 
\hline\hline
\end{tabular}
}
\end{table}
We study the uplink transmission in a mmWave-based dense CRAN cluster~(see~\cref{F:mmWLSysMod}) with $M$ multi-antenna RRHs, denoted by $\mathcal{M}=\{1,\dotsc,M\}$. The cluster is sectorized, where each RRH has $J=3$ sectors covering $120^\circ$ each, and each sector is served by a full-dimensional lens antenna array with $Q$ antenna elements denoted by the set $\mathcal{Q}^{(j)}=\left\{1,\dotsc,Q\right\},\,j\in\mathcal{J}$, where $\mathcal{J}=\{1,\dotsc,J\}$ denotes the set of sectors. 
Further, we use $\mathcal{Q}$ without the superscript to denote the set of all antenna elements in all sectors of an RRH, i.e., $\mathcal{Q}=\bigcup_{j\in\mathcal{J}}\mathcal{Q}^{(j)}$, and $|\mathcal{Q}|=JQ$. Each RRH is connected to the CU via an individual fronthaul link of finite capacity $\bar{R}_m>0$ in bits per second~(bps). There are $K$ single-antenna users in the CRAN cluster, denoted by $\mathcal{K}=\{1,\dotsc,K\}$. All the users and RRHs share the same bandwidth for communication. Depending on a user's location, its signals are incident on at most one sector of an RRH, and we treat any inter-sector/inter-cluster leakage interference as additive Gaussian noise over the bandwidth of interest. We further assume that the total number of sectors in the CRAN is greater than or equal to the number of users, i.e., $JM\geq K$.  
\begin{figure}[t]
\centering
\includegraphics[width=\columnwidth]{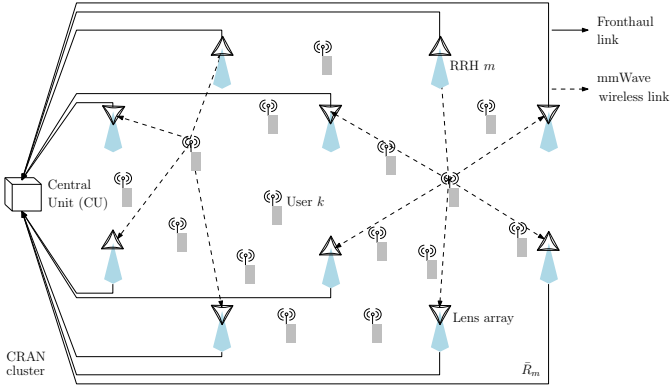}
\caption{Schematic of CRAN.}\label{F:mmWLSysMod}
\end{figure}
The users transmit over a mmWave, frequency-selective, block-fading channel of bandwidth $W$~Hz. The channel between user $k$ and RRH $m$\footnote{Since each user's signals are received by only one sector of every RRH, we refer to the channel between the user and the RRH, instead of the particular sector of the RRH, for convenience.} has $L_{m,k}$ paths denoted by $\mathcal{L}_{m,k}=\{1,\dotsc,L_{m,k}\}$. 
Then the discrete-time channel coefficient vector $\bm h_{m,k}[n]\in\mathbb{C}^{Q\times 1}$ at time index $n$ between user $k$ and RRH $m$ can be expressed using a geometric channel model as 
\begin{align}
\bm h_{m,k}[n]&=\sum_{\ell\in\mathcal{L}_{m,k}}\alpha_{m,k,\ell}\bm a\left(\theta_{m,k,\ell},\phi_{m,k,\ell}\right)\delta[n-d_{m,k,\ell}]\notag\\
&=\sum_{\ell\in\mathcal{L}_{m,k}}\bm h_{m,k,\ell}\delta[n-d_{m,k,\ell}],\label{E:mmWLChanTap}
\end{align}
where $\alpha_{m,k,\ell}\in\mathbb{C}$ and $d_{m,k,\ell}\in\{0,1,\dotsc,d_\mathrm{max}\}\triangleq\mathcal{D}$ denote, respectively, the complex gain and delay in symbol periods corresponding to path $\ell\in\mathcal{L}_{m,k}$, and $\bm a\left(\theta_{m,k,\ell},\phi_{m,k,\ell}\right)\in\mathbb{C}^{Q\times 1}$ is the array response 
for the 
elevation and azimuth angles of arrival of path $\ell$ denoted by $\theta_{m,k,l}$ and $\phi_{m,k,l}$, respectively. 

We consider that for each sector, each RRH is equipped with a rectangular EM lens in the $y-z$ plane, with dimensions $D_y\times D_z$ normalized by the wavelength\footnote{The dimensions are assumed to be same at all RRHs for convenience.} along the $y$- and $z$-axes, respectively. The EM lens is followed by a full-dimensional antenna array with $Q$ elements placed on the \emph{focal surface} of the lens, which is a hemisphere around the lens' center (taken to be the origin in~\cref{F:mmWLLensArray}) with radius equal to the focal length of the lens~(see~\cref{F:mmWLLensArray}). 
Let each antenna element be indexed by a pair of indexes $(q_{\mathrm{e}},q_{\mathrm{a}})$, where $q_\mathrm{e}$ denotes the index in the elevation direction along the focal surface of the lens and $q_\mathrm{a}$ denotes the index in the azimuth direction along the focal surface. Now, for a ray drawn from the center of the lens 
to an antenna element 
$(q_\mathrm{e},q_\mathrm{a})$,  let $\phi_{q_{\mathrm{a}}}\in[-\Phi_{-},\Phi_{+}]$ denote the azimuth angle made by the ray, where $\Phi_-,\Phi_+\in(0,\pi/2]$ are the maximum azimuth angles in the negative and positive $y$-directions, respectively. 
Similarly, let $\theta_{q_{\mathrm{e}}}\in[-\Theta_{-},\Theta_{+}]$ denote the elevation angle made by the ray, 
where $\Theta_{-},\Theta_{+}\in(0,\pi/2]$ are the maximum elevation 
angles covered by the antenna array in the negative and positive $z$-directions, 
respectively. 
Then, the antennas are placed such that the indexes $q_{\mathrm{e}}$ run from the integers $-\lfloor{D_z\sin\Theta_{-}}\rfloor$ to $\lfloor{D_z\sin\Theta_{+}}\rfloor$. Thus, the elevation angles of the antenna elements $\theta_{q_\mathrm{e}}$ are related to the indexes $q_\mathrm{e}$ by 
$\sin\theta_{q_{\mathrm{e}}}=q_{\mathrm{e}}/D_z,\,q_{\mathrm{e}}\in\{-\lfloor{D_z\sin\Theta_{-}}\rfloor,\dotsc,-1,0,1,\dotsc,\lfloor{D_z\sin\Theta_{+}}\rfloor\}$. 
Next, for each 
index $q_{\mathrm{e}}$, the index $q_\mathrm{a}$ runs from the integers $-\lfloor D_y\cos\theta_{q_{\mathrm{e}}}\sin\Phi_{-}\rfloor$ to $\lfloor D_y\cos\theta_{q_{\mathrm{e}}}\sin\Phi\rfloor$, so that the azimuth angles of the antenna elements $\phi_{q_\mathrm{a}}$ are related to the indexes $q_\mathrm{a}$ by $\sin\phi_{q_{\mathrm{a}}}=q_{\mathrm{a}}/(D_y\cos\theta_{q_{\mathrm{e}}}),\,q_{\mathrm{a}}\in\{-\lfloor D_y\cos\theta_{q_{\mathrm{e}}}\sin\Phi_{-}\rfloor,\dotsc,-1,0,1,\dotsc,\lfloor D_y\cos\theta_{q_{\mathrm{e}}}\sin\Phi\rfloor\}$\cite{zeng-zhang2017cost}. 
Then 
the amplitude response in~\eqref{E:mmWLChanTap} of the lens array element $\left(q_{\mathrm{e}},q_{\mathrm{a}}\right)\equiv q\in\mathcal{Q}^{(j)}$, to a uniform plane wave incident at elevation and azimuth angles $(\theta,\phi)$, can be expressed as~\cite{zeng-etal2016multiuser}
\begin{align}
a_q\left(\theta,\phi\right)&=\sqrt{D_zD_y}\sinc\left(q_{\mathrm{e}}-D_z\sin\theta\right)\notag\\
&\quad\cdot\sinc\left(q_{\mathrm{a}}-D_y\cos\theta\sin\phi\right).
\label{E:mmWLLARes}
\end{align}
Since the RRHs have fronthaul links of finite capacities, they must quantize the received signals before forwarding to the CU. Moreover, the RRHs are typically low-cost nodes with limited processing capability. Thus, we consider that they perform uniform SQ independently on each antenna, with 
the SQ bit allocation adapted in each channel coherence interval. 
As shown in~\cref{F:mmWLSPRRH}, the RRHs perform bit allocation based on the estimated received power levels at each antenna element, which can be obtained before converting the signals to the baseband, either using feedback from the automatic gain control~(AGC) circuitry, or by means of analog power estimators~\cite{narasimhamurthy-tepedelenlioglu2009antenna}, which can be implemented using band-pass filters and envelope detectors. 
\begin{figure}[t]
\centering
\includegraphics[width=\columnwidth]{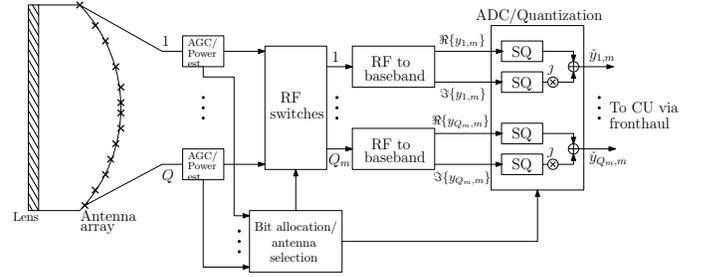}
\caption{Architecture of RRH with lens antenna array. 
}\label{F:mmWLSPRRH}
\end{figure}
We consider a frame-based transmission by the users, with frame duration $T_\mathrm{f}<T_\mathrm{c}$ in symbol periods, where $T_\mathrm{c}$ denotes the minimum coherence time among all the user-RRH channels. 
Each frame is further divided into the following three stages as shown in~\cref{F:mmWLFrameStruct}:
\begin{itemize}
\item A power probing stage of duration $T_\mathrm{a}$, where the users transmit constant amplitude signals in order to enable the RRHs to perform bit allocation and antenna selection;
\item A channel training stage of duration $T_\mathrm{p}$, where the users transmit pilot symbols, and the CU performs channel estimation for the selected antennas using the quantized signals forwarded by the RRHs; and 
\item A data transmission stage of duration $T_\mathrm{d}$, where the users transmit their data, which is quantized and forwarded by the RRHs for decoding at the CU.
\end{itemize} 
The above stages are separated by guard intervals of $d_\mathrm{max}$ symbols~(see~\cref{F:mmWLFrameStruct}), where $d_\mathrm{max}$ denotes the maximum delay spread of all the user-RRH channels. Note that $T_\mathrm{f}=T_\mathrm{a}+T_\mathrm{p}+T_\mathrm{d}+2d_\mathrm{max}$. In the following, we describe each stage in detail. 

\subsection{Uniform Scalar Quantization~(SQ) and Bit Allocation at RRHs}\label{SS:mmWLRRHProc}
During the power-probing stage, the users transmit constant amplitude signals for a duration $T_\mathrm{a}\geq d_\mathrm{max}+1$, and each RRH $m$ obtains an estimate of the average received power $\rho_{q,m}$ on each antenna either from the AGC circuitry, or using analog power estimators. We consider that each RRH performs uniform SQ independently on the real and imaginary components of the complex baseband samples $y_{q,m}[n]$ received at an antenna $q$, using $b_{q,m}\in\mathbb{Z}_{+}$ bits.\footnote{This can be performed using low-cost, low-resolution analog-to-digital converters (ADCs).}  
If $b_{q,m}=0$, the symbols on antenna $q$ are not forwarded to the CU by RRH $m$~(i.e., this antenna is not selected for subsequent channel training and data transmission). 
Following the 
design in~\cite{liu-etal2015joint}, 
the resulting quantized samples $\check{y}_{q,m}[n]$ can 
be expressed as 
\begin{align}
\check{y}_{q,m}[n]=y_{q,m}[n]+e_{q,m}[n],
\label{E:mmWLQSig}
\end{align}
where $e_{q,m}[n]$ represents the quantization error, modeled as a random variable with mean zero and variance given by~\cite{liu-etal2015joint} 
\begin{align}
\varepsilon^2_{q,m}\triangleq\mathbb{E}[|e_{q,m}[n]|^2]=3\rho_{q,m}/2^{2b_{q,m}},\quad b_{q,m}\in\mathbb{Z}_{++}.\label{E:mmWLQNoiseV}
\end{align} 
The quantization error is assumed to be uncorrelated with $y_{q,m}[n]$, and as the SQ is performed independently at each antenna and for each sample, we have 
$\mathbb{E}[e_{q,m}[n]e^*_{q',m}[n']]=0$ for any $q'\neq q$ or $n'\neq n$, $\forall m\in\mathcal{M}$. 
With Nyquist rate sampling, the transmission rate required 
to forward the quantized signals 
over all the antennas is $\sum_{q\in\mathcal{Q}}2Wb_{q,m}$ bps, which must not exceed the 
fronthaul capacity $\bar{R}_m$. 
\begin{figure}[t]
\centering
\includegraphics[width=\columnwidth]{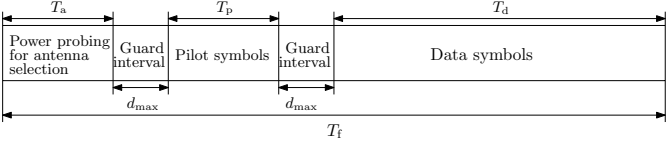}
\caption{Frame structure for uplink transmission with lens antenna arrays.}\label{F:mmWLFrameStruct}
\end{figure}
Each RRH $m$ 
uses the estimate $\rho_{q,m}$ of the average received power computed from the received signals in the power-probing stage to perform the SQ in the subsequent channel training and data transmission stages. We consider the design of quantization bit allocation 
to minimize the 
total SQ noise power over all the antennas subject to the fronthaul capacity constraint at the RRH, as captured by the following optimization problem\footnote{We assume that the objective function is defined for $b_{q,m}=0$ as well.} 
\begin{subequations}
\label{P:BA}
\begin{align}
\mathop{\mathrm{minimize}}_{\bm b_m\in\mathbb{Z}^{JQ\times 1}_+}&\,\sum_{q\in\mathcal{Q}}\frac{3\rho_{q,m}}{2^{2b_{q,m}}}\tag{\ref*{P:BA}}\\
\mathrm{subject}~\mathrm{to}&\notag\\
&\,\sum_{q\in\mathcal{Q}}b_{q,m}\leq \frac{\bar{R}_m}{2W}.\label{C:FHR}
\end{align}
\end{subequations}
Since the $b_{q,m}$'s are integers, the above problem is non-convex. However, if the variables are relaxed so that $b_{q,m}\in\mathbb{R}_+\,\forall q\in\mathcal{Q}$, we have the following relaxed problem 
\begin{subequations}
\label{P:BARel}
\begin{align}
\min_{\bm b_m\in\mathbb{R}_+^{JQ\times 1}}&\,\sum_{q\in\mathcal{Q}}\frac{3\rho_{q,m}}{2^{2b_{q,m}}}\tag{\ref*{P:BARel}}\\
\mathrm{s.t.}&\,\cref{C:FHR}\notag
\end{align}
\end{subequations}
which is convex since the objective function is convex, 
while the constraint~\eqref{C:FHR} is linear. Thus, we have the following proposition.   
\begin{proposition}\label{Prop:mmWLOptBARel}
The optimal solution to problem~\eqref{P:BARel} is given by 
\begin{align}
b'_{q,m}=\max\Big\{\frac{1}{2}\log_2\Big(\frac{6\rho_{q,m}\ln 2}{\lambda^\star}\Big),0\Big\},\,q\in\mathcal{Q},\label{E:mmWLBAMinLag}
\end{align}
where $\lambda^\star\geq 0$ is such that $\sum_{q\in\mathcal{Q}}b'_{q,m}=\bar{R}_m/2W$.
\end{proposition}
\begin{proof}
Please refer to the appendix.
\end{proof}
Let $\mathcal{Q}'_m$ denote the set of antennas with non-zero allocation $b'_{q,m}$ according to~\eqref{E:mmWLBAMinLag}. 
With the optimal solution $\bm b'_m$ to problem~\eqref{P:BARel}, we proceed to construct a feasible integer solution $\tilde{\bm b}_m\in\mathbb{Z}_+^{JQ\times 1}$ for the original problem~\eqref{P:BA} by rounding $\bm b'_m$ as follows\footnote{We construct a feasible integer solution $\tilde{\bm b}_m$ such that if $b'_{q,m}=0$ for some $q\in\mathcal{Q}$ in~\eqref{E:mmWLBAMinLag}, then $\tilde{b}_{q,m}=0$ as well.}
\begin{align}
\tilde{b}_{q,m}=\begin{cases}
\lfloor b'_{q,m}\rfloor&\text{if }b'_{q,m}-\lfloor b'_{q,m}\rfloor\leq\beta\\
\lceil b'_{q,m}\rceil&\text{if }b'_{q,m}-\lfloor b'_{q,m}\rfloor>\beta\label{E:mmWLRoundBA}
\end{cases},\,q\in\mathcal{Q}'_m,
\end{align}
where $\beta\in[0,1]$ is an appropriate threshold. Notice that as $\beta$ is decreased, most of the $b'_{q,m}$'s would be rounded above, making it more difficult to satisfy the constraint in~\eqref{P:BA} and vice versa. Hence, a suitable $\beta$ can be found by bisection on the interval $[0,1]$, each time evaluating the constraint in~\eqref{P:BA} with $\tilde{\bm b}_m$ in~\eqref{E:mmWLRoundBA} and updating $\beta$ accordingly. 

From~\eqref{E:mmWLBAMinLag}, we observe that the optimal solution to problem~\eqref{P:BARel} allocates more bits to the antenna with higher estimated power $\rho_{q,m}$, which is desirable, since the antennas that receive stronger signals from the users are more likely to have useful information to be decoded. 
The algorithm is summarized in~\cref{A:BA}, and involves two instances of bisection search. The number of iterations in the first instance to find $\lambda^\star$ is $\log_2(\tfrac{6\rho_{\mathrm{max},m}\ln 2}{\epsilon})$, and the number of iterations in the second instance to find a suitable $\beta$ is $\log_2(1/\epsilon)$. The total number of iterations for the algorithm is thus $\log_2(\tfrac{6\rho_{\mathrm{max},m}\ln 2}{\epsilon^2})$. If all the $\rho_{q,m}$'s are normalized such that $6\rho_{\mathrm{max},m}\ln 2=1$, and $\epsilon\ll 1$, the total number of iterations is $2\log_2(1/\epsilon)$, which is negligible for typical values of $\epsilon$, like $10^{-9}$. Thus, based on the received power estimates in the power probing stage, each RRH $m$ selects the set of antennas denoted by $\mathcal{Q}_m\triangleq\{q\in\mathcal{Q}\mid \tilde{b}_{q,m}>0\}$ with $|\mathcal{Q}_m|\triangleq Q_m$ and forwards the received symbols on these antennas to the CU after SQ in the subsequent channel training and data transmission stages. We assume that the fronthaul capacity at each RRH is such that $\bar{R}_m\geq 2JW$, so that the number of selected antennas $Q_m$ is at least $J$, i.e.\ $Q_m\geq J,\,\forall m\in\mathcal{M}$. Then, $Q_\mathrm{tot}\geq JM$, and since $JM\geq K$, it is always feasible to recover all $K$ users' signals at the CU via linear processing with independent channels. Let $Q_\mathrm{tot}\triangleq\sum_{m\in\mathcal{M}}Q_m$ denote the total number of selected antennas in the network. For the ease of exposition, we refer to the selected antennas as ``streams" and use the index $i\in\{1,\dotsc,Q_\mathrm{tot}\}\triangleq\mathcal{I}$ to denote each stream so that the stream $i$ corresponds to the antenna $q_i^{(j)}$ in sector $j$ of RRH $m_i$. 
\begin{table}[t]
\centering
\caption{Algorithm for quantization bit allocation}\label{A:BA}
\vspace{-0.5cm}
\begin{framed}
\begin{algorithmic}[1]
\State Initialize 
$\lambda_{\mathrm{min}}=0$, $\lambda_{\mathrm{max}}=6\rho_{\mathrm{max},m}\ln 2$, tolerance $\epsilon>0$
\Repeat
\State $\lambda=(\lambda_{\mathrm{min}}+\lambda_{\mathrm{max}})/2$
\State Compute $\bm b_m$ according to~\eqref{E:mmWLBAMinLag}
\If {$\sum_{q\in\mathcal{Q}}b_{q,m}\leq\bar{R}_m/2W$}
\State Set $\lambda_{\mathrm{max}}=\lambda$
\Else
\State Set $\lambda_{\mathrm{min}}=\lambda$
\EndIf
\Until $|\lambda_{\mathrm{max}}-\lambda_{\mathrm{min}}|<\epsilon$
\State Use $\lambda^\star=\lambda$ in~\eqref{E:mmWLBAMinLag} to compute 
$\bm b'_m$
\State Initialize $\beta_{\mathrm{min}}=0$, $\beta_{\mathrm{max}}=1$, and find $\beta$ in~\eqref{E:mmWLRoundBA} by bisection similar to above until $|\beta_{\mathrm{max}}-\beta_{\mathrm{min}}|<\epsilon$
\State Compute final integer solution $\tilde{\bm b}_m$ to problem~\eqref{P:BA} according to~\eqref{E:mmWLRoundBA} with converged $\beta$ 
\end{algorithmic}
\end{framed}
\end{table}
\subsection{Path Delay Compensation and Achievable Sum-Rate with Perfect CSI}\label{SS:mmWLPDCPerfCSILens}
In this subsection, we describe the operations at the CU, assuming perfect CSI. If the RRHs are equipped with lens antenna arrays, we propose that the users transmit simultaneously with single-carrier modulation, so that the data symbols transmitted by user $k$ are given by  $x_{\mathrm{d},k}[n]=\sqrt{P}s_{\mathrm{d},k}[n]$, where $P\geq 0$ is the transmit power 
and $s_{\mathrm{d},k}[n]\sim\mathcal{CN}(0,1)$ is the complex data symbol of user $k$. 
Using~\cref{E:mmWLChanTap,E:mmWLQSig}, the quantized symbols received at the CU corresponding to stream $i\in\mathcal{I}$ can be expressed as
\begin{align}
\check{y}_{\mathrm{d},i}[n]&=\sum_{k\in\mathcal{K}}\sum_{\ell\in\mathcal{L}_{m_i,k}} h_{i,k,\ell}x_{\mathrm{d},k}[n-d_{m_i,k,\ell}]+z_{\mathrm{d},i}[n]
+e_{\mathrm{d},i}[n].
\end{align}
We assume that the CU decodes the users' data symbols via linear processing, after path delay compensation~\cite{zeng-etal2016multiuser} on the received signals for each stream as described below. Let 
$\ell^\star_{i,k}=\mathop{\arg\max}_{\ell\in\mathcal{L}_{m_i,k}}|h_{i,k,\ell}|^2$ denote the strongest path among all those arriving on stream~(selected antenna) $i\in\mathcal{I}$ from user $k$. Due to the response of the lens array, the delayed versions of a particular user's signals with different angles of arrival are focused on different antenna elements at each RRH in general. Thus, to ensure that the symbols of each user that have undergone the strongest path gain on each stream are combined at the CU in a synchronized manner, each stream of quantized symbols $\check{y}_{\mathrm{d},i}[n]$, corresponding to antenna $q_i^{(j)}$ of RRH $m_i$, is advanced by the delay $d_{m_i,k,\ell^\star_{i,k}}$ corresponding to the path $\ell^\star_{i,k}$, to obtain the \emph{delay compensated} signal $\check{y}_{\mathrm{d},i,k}\big[n+d_{m_i,k,\ell^\star_{i,k}}\big]\triangleq\bar{\check{y}}_{\mathrm{d},i,k}[n]$ for each user $k\in\mathcal{K}$, as given by
\begin{align}
&\bar{\check{y}}_{\mathrm{d},i,k}[n]\notag\\
&=\sum_{k'\in\mathcal{K}}\sum_{\ell'\in\mathcal{L}_{m_i,k'}} h_{i,k',\ell'}x_{\mathrm{d},k}\big[n-\big(d_{m_i,k',\ell'}-d_{m_i,k,\ell^\star_{i,k}}\big)\big]\notag\\&\quad+z_{\mathrm{d},i}\big[n+d_{m_i,k,\ell^\star_{i,k}}\big]+e_{\mathrm{d},i}\big[n+d_{m_i,k,\ell^\star_{i,k}}\big].\label{E:mmWLRxQSigCUAdv}
\end{align} 
In order to write the 
summation over the paths in~\eqref{E:mmWLRxQSigCUAdv} in terms of the \emph{delay differences} $\nu=\big(d_{m_i,k',\ell'}-d_{m_i,k,\ell^\star_{i,k}}\big)\in\{0,\pm 1,\dotsc,\pm d_\mathrm{max}\}\triangleq\Delta$ with the delay of the maximum gain path $d_{m_i,k,\ell^\star_{i,k}}$, we define for each antenna 
$i\in\mathcal{I}$ and user pair $k,k'\in\mathcal{K}$, the new channel coefficient 
\begin{align}
\bar{h}_{i,kk'}[\nu]
&\triangleq\sum_{\ell'\in\mathcal{L}_{m_i,k'}}h_{i,k',\ell'}\delta\big[\nu-\big(d_{m_i,k',\ell'}-d_{m_i,k,\ell^\star_{i,k}}\big)\big],\notag\\
&\quad i\in\mathcal{I},k,k'\in\mathcal{K},\nu\in\Delta,\label{E:mmWLDCChanCoeff}
\end{align}
which is equivalent to the channel coefficient~(or sum of coefficients) corresponding to the path(s) $\ell'\in\mathcal{L}_{m_i,k'}$ from user $k'$ to antenna $i$, which has~(have) a delay difference of $\nu$ 
with the maximum gain path $\ell^\star_{i,k}$ of user $k$ to the same antenna. Then,~\eqref{E:mmWLRxQSigCUAdv} can be expressed as
\begin{align}
\bar{\check{y}}_{\mathrm{d},i,k}[n]&=\bar{h}_{i,kk}[0]x_{\mathrm{d},k}[n]+\sum_{\nu\in\Delta\setminus\{0\}} \bar{h}_{i,kk}[\nu]x_{\mathrm{d},k}[n-\nu]\notag\\
&\quad+\sum_{k'\in\mathcal{K}\setminus\{k\}}\sum_{\nu\in\Delta}\bar{h}_{i,kk'}[\nu]x_{\mathrm{d},k'}[n-\nu]\notag\\
&\quad+\bar{z}_{\mathrm{d},i,k}[n]+\bar{e}_{\mathrm{d},i,k}[n],\,i\in\mathcal{I},
k\in\mathcal{K},\label{E:mmWLDCSigCU}
\end{align}
where $\bar{z}_{\mathrm{d},i,k}[n]\triangleq z_{\mathrm{d},i}\big[n+d_{m_i,k,\ell^\star_{i,k}}\big]\sim\mathcal{CN}(0,\sigma_{m_i}^2)$ and $\bar{e}_{\mathrm{d},i,k}[n]\triangleq e_{\mathrm{d},i}\big[n+d_{m_i,k,\ell^\star_{i,k}}\big]$, with $\mathbb{E}[\bar{e}_{\mathrm{d},i,k}[n]]=0$  and $\mathbb{E}[|\bar{e}_{\mathrm{d},i,k}[n]|^2]=\varepsilon_i^2$ defined in~\eqref{E:mmWLQNoiseV} denote the AWGN and quantization noise samples shifted by $d_{m_i,k,\ell^\star_{i,k}}$ symbol periods. Note that $\bar{\check{y}}_{\mathrm{d},i,k}[n]$ in~\eqref{E:mmWLDCSigCU} depends on user $k$ whose maximum gain path $\ell^\star_{i,k}$ is used as reference. The second and third terms in~\eqref{E:mmWLDCSigCU} represent user $k$'s own delayed symbols, and the interfering symbols from other users, respectively. Collecting the signals from all the 
streams, \eqref{E:mmWLDCSigCU} can be written in vector form as
\begin{align}
\bar{\check{\bm y}}_{\mathrm{d},k}[n]&=\bar{\bm h}_{kk}[0]x_{\mathrm{d},k}[n]+\sum_{\nu\in\Delta\setminus\{0\}} \bar{\bm h}_{kk}[\nu]x_{\mathrm{d},k}[n-\nu]\notag\\
&\quad+\sum_{k'\in\mathcal{K}\setminus k} \sum_{\nu\in\Delta}\bar{\bm h}_{kk'}[\nu]x_{\mathrm{d},k'}[n-\nu]+\bar{\bm z}_{\mathrm{d},k}[n]\notag\\
&\quad+\bar{\bm e}_{\mathrm{d},k}[n],
\quad k\in\mathcal{K}.
\label{E:mmWLDCSigCUV}
\end{align}
where all the vectors are of dimension $Q_\mathrm{tot}\times 1$. We consider that the CU performs linear receive beamforming on $\bar{\check{\bm y}}_k[n]$ with the beamforming vector $\bm u_k\in\mathbb{C}^{Q_{\mathrm{tot}}\times 1}$ 
to construct the estimate $\hat{x}_{\mathrm{d},k}[n]=\bm u^\mathsf{H}_k\bar{\check{\bm y}}_{\mathrm{d},k}[n]$ of user $k$'s symbol. Treating the inter-symbol and inter-user interference in~\eqref{E:mmWLDCSigCU} as Gaussian noise, the signal to interference-plus-noise ratio~(SINR) for decoding $x_{\mathrm{d},k}[n]$ is given by~\cref{E:mmWLSNRUkL} on the top of the next page.  
\begin{floatEq}
\begin{align}
\gamma_k=\frac{P\left|\bm u^\mathsf{H}_k\bar{\bm h}_{kk}[0]\right|^2}{\sum_{\nu\in\Delta\setminus\{0\}}P|\bm u^\mathsf{H}_k\bar{\bm h}_{kk}[\nu]|^2+\sum_{k'\in\mathcal{K}\setminus k} \sum_{\nu\in\Delta}P|\bm u^\mathsf{H}_k\bar{\bm h}_{kk'}[\nu]|^2+\bm u^\mathsf{H}_k(\bm\Sigma+\bm\Xi)\bm u_k},\label{E:mmWLSNRUkL}
\end{align}
\end{floatEq}
In~\eqref{E:mmWLSNRUkL}, $\bm\Sigma\triangleq\mathrm{blkdiag}\begin{psmallmatrix}
\sigma_1^2\bm I_{Q_1}&\cdots&\sigma_M^2\bm I_{Q_{M}}
\end{psmallmatrix}\in\mathbb{C}^{Q_\mathrm{tot}\times Q_\mathrm{tot}}$, 
and $\bm \Xi\triangleq\mathrm{diag}
\big(\{\varepsilon_{i}^2\}_{i=1}^{Q_\mathrm{tot}}\big)\in\mathbb{C}^{Q_\mathrm{tot}\times Q_\mathrm{tot}}$,  
where $\varepsilon^2_{i},\,i\in\mathcal{I}$, is defined in~\eqref{E:mmWLQNoiseV}. 
Since the transmit powers of the users are fixed, $\gamma_k$ in~\eqref{E:mmWLSNRUkL} 
is maximized by $\bm u_k=\bm C_k^{-1}\bar{\bm h}_{kk}[0],\,k\in\mathcal{K}$ according to the minimum mean squared error~(MMSE) criterion~\cite{tse2005fundamentals}, where 
\begin{align}
\bm C_k&\triangleq\sum_{\nu\in\Delta\setminus\{0\}}P\bar{\bm h}_{kk}[\nu]\bar{\bm h}_{kk}[\nu]^\mathsf{H}\notag\\
&\quad+\sum_{k'\in\mathcal{K}\setminus\{k\}}\sum_{\nu\in\Delta}P\bar{\bm h}_{kk'}[\nu]\bar{\bm h}_{kk'}[\nu]^\mathsf{H}+\bm\Sigma+\bm \Xi,\,k\in\mathcal{K},\label{E:mmWLMMSEC}
\end{align}
denotes the covariance of the noise and interference terms in~\eqref{E:mmWLDCSigCUV}. In this case, the SINR in~\eqref{E:mmWLSNRUkL} becomes $\gamma_k=P\bar{\bm h}_{kk}[0]^\mathsf{H}\bm C_k^{-1}\bar{\bm h}_{kk}[0]$, and by assuming the worst-case CSCG distribution for the quantization noise, a lower bound on the achievable sum rate over all users is thus given by 
\begin{align}
r_\mathrm{lens}=\sum_{k\in\mathcal{K}}\log_2(1+\gamma_k)
\end{align}
in bps/Hz. In the next subsection, we consider channel estimation at the CU, and extend the achievable rate analysis to the case with imperfect CSI.

\subsection{Channel Estimation and Achievable Sum-Rate with Imperfect CSI}\label{SS:mmWLImperfCSILens}
Since the CU does not know the CSI \emph{a priori}, it needs to estimate the CSI using pilot signals sent by the users and then quantized and forwarded by the RRHs in the channel estimation stage (see~\cref{F:mmWLFrameStruct}). The users transmit known pilot symbols given by $x_{\mathrm{p},k}[n]=\sqrt{P}s_{\mathrm{p},k}[n],\,n=0,1,\dotsc,T_\mathrm{p}-1,k\in\mathcal{K}$. 
Let $\bm h_{i,k}\triangleq\begin{bsmallmatrix}h_{i,k}[0]&\cdots & h_{i,k}[d_\mathrm{max}]\end{bsmallmatrix}^\mathsf{T}\in\mathbb{C}^{(d_\mathrm{max}+1)\times 1}$ denote the vector of time-domain channel taps from user $k$ to antenna $i\in\mathcal{I}$. 
Then, the vector of quantized symbols $\check{\bm y}_{\mathrm{p},i}\triangleq\begin{bsmallmatrix} \check{y}_{\mathrm{p},i}[0] &\cdots &\check{y}_{\mathrm{p},i}[T_\mathrm{p}-1]\end{bsmallmatrix}^\mathsf{T}\in\mathbb{C}^{T_\mathrm{p}\times 1}$ 
received at the CU during the channel estimation stage can be expressed as
\begin{align}
\check{\bm y}_{\mathrm{p},i}&=\sum_{k\in\mathcal{K}}\bm X_{\mathrm{p},k}\bm h_{i,k}+\bm z_{\mathrm{p},i}+\bm e_{\mathrm{p},i}\notag\\
&=\bm X_\mathrm{p}\bm h_i+\bm z_{\mathrm{p},i}+\bm e_{\mathrm{p},i},\,i\in\mathcal{I},
\label{E:mmWLRxSigQCUCELensV}
\end{align}
where 
\begin{align}
\bm X_{\mathrm{p},k}&\triangleq\begin{bsmallmatrix}
x_{\mathrm{p},k}[0] & 0 & \cdots & 0\\
x_{\mathrm{p},k}[1] & x_{\mathrm{p},k}[0] & \cdots & 0\\
\vdots &\vdots & \ddots &\vdots &\\
x_{\mathrm{p},k}[T_\mathrm{p}-1] & x_{\mathrm{p},k}[T_\mathrm{p}-2] &\cdots & x_{\mathrm{p},k}[T_\mathrm{p}-1-d_\mathrm{max}]\end{bsmallmatrix}\notag\\
&\in\mathbb{C}^{T_\mathrm{p}\times (d_\mathrm{max}+1)},\quad k\in\mathcal{K},\label{E:mmWLPSUsrMatLens}
\end{align}
is a Toeplitz matrix constructed from consecutive shifts of the pilot symbols of the users, 
while $\bm z_{\mathrm{p},i}\triangleq\begin{bsmallmatrix}z_{\mathrm{p},i}[0]&\cdots& z_{\mathrm{p},i}[T_\mathrm{p}-1]\end{bsmallmatrix}^\mathsf{T}$ 
with $\bm z_{\mathrm{p},i}\sim\mathcal{CN}(\bm 0,\sigma_{m_i}^2\bm I_{T_\mathrm{p}})$ denotes the AWGN, and $\bm e_{\mathrm{p},i}\triangleq\begin{bsmallmatrix}e_{\mathrm{p},i}[0]&\cdots& e_{\mathrm{p},i}[T_\mathrm{p}-1]\end{bsmallmatrix}^\mathsf{T}\in\mathbb{C}^{T_\mathrm{p}\times 1}$ 
with $\mathbb{E}[\bm e_{\mathrm{p},i}\bm e_{\mathrm{p},i}^\mathsf{H}]=\varepsilon_i^2\bm I_{T_\mathrm{p}}$ denotes the quantization noise. 
Also, $\bm X_\mathrm{p}\triangleq\begin{bsmallmatrix}\bm X_{\mathrm{p},1}&\cdots &\bm X_{\mathrm{p},K}\end{bsmallmatrix}\in\mathbb{C}^{T_\mathrm{p}\times K(d_\mathrm{max}+1)}$, 
and $\bm h_i\triangleq\begin{bsmallmatrix}\bm h_{i,1}^\mathsf{T}&\cdots &\bm h_{i,K}^\mathsf{T}\end{bsmallmatrix}^\mathsf{T}\in\mathbb{C}^{K(d_\mathrm{max}+1)\times 1}$ in~\eqref{E:mmWLRxSigQCUCELensV}. 
It is sufficient to estimate 
$\bm h_i,\,i\in\mathcal{I}$ in~\eqref{E:mmWLPSUsrMatLens}, to find both the channel coefficients and their corresponding tap delays. We do not assume any prior knowledge of the probability distribution function (pdf) of the elements of 
$\bm h_i,i\in\mathcal{I}
$ at the CU, i.e. they are treated as unknown constants. 
Then, the least-squares~(LS) estimate of $\bm h_i$ is given by 
\begin{align}
\hat{\bm h}_i&=\mathop{\arg\min}_{\bm h_i}\|\check{\bm y}_{\mathrm{p},i}-\bm X_\mathrm{p}\bm h_i\|^2\notag\\
&=\bm X_\mathrm{p}^\dag\check{\bm y}_{\mathrm{p},i}=\bm h_i+
\bm X_\mathrm{p}^\dag(\bm z_{\mathrm{p},i}+\bm e_{\mathrm{p},i}),\,i\in\mathcal{I}.
\label{E:mmWLLSEChan}
\end{align}
where we define $\bm\xi_{\mathrm{p},i}\triangleq\bm X_{\mathrm{p}}^\dag\bm z_{\mathrm{p},i}+\bm X_{\mathrm{p}}^\dag\bm e_{\mathrm{p},i}$. Since each component of the vector $\bm X_{\mathrm{p}}^\dag\bm e_{\mathrm{p},i}\in\mathbb{C}^{(d_\mathrm{max}+1)\times 1}$ is a linear combination of the $T_\mathrm{p}$ independent zero-mean random variables in $\bm e_{\mathrm{p},i}$, the second term in $\bm\xi_{\mathrm{p},i}$, 
$\bm X_{\mathrm{p}}^\dag\bm e_{\mathrm{p},i}$, can be modeled as a CSCG random vector, with mean $\mathbb{E}[\bm X_{\mathrm{p}}^\dag\bm e_{\mathrm{p},i}]=\bm 0$ and covariance $\mathbb{E}[\bm X_{\mathrm{p}}^\dag\bm e_{\mathrm{p},i}\bm e_{\mathrm{p},i}^\mathsf{H}(\bm X_{\mathrm{p}}^\dag)^\mathsf{H}]=\varepsilon_i^2(\bm X_\mathrm{p}^\mathsf{H}\bm X_\mathrm{p})^{-1}$, due to the central limit theorem. Thus, $\bm\xi_{\mathrm{p},i}\sim\mathcal{CN}\big(\bm 0,(\sigma_{m_i}^2+\varepsilon_i^2)(\bm X_\mathrm{p}^\mathsf{H}\bm X_\mathrm{p})^{-1}\big)$ in~\eqref{E:mmWLLSEChan}. Note that $T_\mathrm{p}\geq K(d_\mathrm{max}+1)$ must be satisfied 
for the solution in~\eqref{E:mmWLLSEChan} to exist. The MSE of the estimate $\hat{\bm h}_i$ is given by
\begin{align}
\mathbb{E}[\|\hat{\bm h}_i-\bm h_i\|^2]=\mathbb{E}[\|\bm\xi_{\mathrm{p},i}\|^2]=\big(\sigma_{m_i}^2+\varepsilon_i^2\big)\mathrm{tr}\big((\bm X_\mathrm{p}^\mathsf{H}\bm X_\mathrm{p})^{-1}\big),\label{E:mmWLCEMSELens}
\end{align}
which is minimized if  
$\bm X_\mathrm{p}^\mathsf{H}\bm X_\mathrm{p}=c\bm I_{K(d_\mathrm{max}+1)}
$, where $c$ is a constant~\cite[Example~4.3]{kay1993fundamentals}. From the construction of $\bm X_\mathrm{p}$, this translates to the condition 
$\bm X_{\mathrm{p},k}^\mathsf{H}\bm X_{\mathrm{p},k'}=c\delta[k-k']\bm I_{(d_\mathrm{max}+1)},\forall k,k'\in\mathcal{K}$, which is satisfied if each user's training sequence $s_{\mathrm{p},k}[n],\,n=0,\dotsc,T_\mathrm{p}-1$ is orthogonal to that of every other user's training sequence, and has the ``ideal" auto-correlation property~\cite{kay1993fundamentals}. This can be ensured by using e.g., unit-amplitude Zadoff-Chu sequences~\cite{popovic1992generalized,sesia2009lte} for $s_{\mathrm{p},k}[n]$, with equal transmit power $P,\forall k\in\mathcal{K}$. 
In this case, $c=PT_\mathrm{p}$ so that $(\bm X_{\mathrm{p},k}^\mathsf{H}\bm X_{\mathrm{p},k})^{-1}=\tfrac{1}{PT_\mathrm{p}}\bm I_{(d_\mathrm{max}+1)},\,\forall k\in\mathcal{K}$, and thus $\bm X_\mathrm{p}^\dag=\tfrac{1}{PT_\mathrm{p}}\bm X_\mathrm{p}^\mathsf{H}$. Consequently, the LS estimate in~\eqref{E:mmWLLSEChan} reduces to 
\begin{align}
\hat{\bm h}_i=\frac{1}{PT_\mathrm{p}}\bm X_\mathrm{p}^\mathsf{H}\check{\bm y}_{\mathrm{p},i}=\bm h_i+\bm\xi_{\mathrm{p},i},\label{E:mmWLLSEChanO}
\end{align} 
where $\bm\xi_{\mathrm{p},i}=\frac{1}{PT_\mathrm{p}}\bm X_\mathrm{p}^\mathsf{H}(\bm z_{\mathrm{p},i}+\bm e_{\mathrm{p},i})\sim\mathcal{CN}\Big(\bm 0,\textit{•}\tfrac{(\sigma_{m_i}^2+\varepsilon_i^2)}{PT_\mathrm{p}}\bm I_{K(d_\mathrm{max}+1)}\Big)$. According to~\eqref{E:mmWLLSEChanO}, each tap in $\bm h_i$ is estimated by correlating the received signal vector $\check{\bm y}_{\mathrm{p},i}$ with the training sequence of the corresponding user shifted by the corresponding delay of the tap. 

Due to the energy focusing property of the lens antenna array in~\eqref{E:mmWLLARes}, and the different directions of arrival of the users' signals at different RRHs, at any stream $i$, typically only one path corresponding to a particular user, and corresponding to a particular tap delay, would dominate over all its other paths. 
Thus, the magnitudes of the channel coefficients in $\bm h_i$ can be vastly different, depending on the angle of arrival of the users' signals, and directly using the estimates of all the taps in $\bm h_i,\,i\in\mathcal{I}$ according to~\eqref{E:mmWLLSEChan} may be ineffective. To overcome this issue, we propose a \emph{reduced-size, approximate} linear MMSE beamforming by exploiting the ``sinc"-type response of the lens array, where for each user, we select only the streams which contain at least one dominant estimated tap of its own, and then perform an approximate linear MMSE beamforming over these selected streams by thresholding the channel estimates, as explained below. 

Notice that for the data decoding with path delay compensation detailed in~\cref{SS:mmWLPDCPerfCSILens}, the term $\bar{h}_{i,kk}[0]$ in~\eqref{E:mmWLDCSigCU} corresponds to the channel coefficient with maximum path gain from user $k$ to antenna $i$. Now, due to the energy focusing property of the lens, the angle of arrival of user $k$'s signal is typically such that $|\bar{h}_{i,kk}[0]|\gg|\bar{h}_{i,kk}[\nu]|,\,\forall\nu\in\Delta\setminus\{0\}$, on some streams $i$. This means that the maximum gain path of the user would dominate over the other paths on a particular stream $i$, provided the antenna corresponding to stream $i$ has a response which peaks at one of the angles of arrival of the user's paths. Since we do not know the actual channel gains, we aim to find the set of such streams $\mathcal{I}_k$ for each user $k$, for which the magnitude of the maximum \emph{estimated} channel gain \emph{exceeds} a certain threshold. For this, we first estimate the tap delays of the maximum gain paths from user $k$ to every stream $i\in\mathcal{I}$. These are denoted by $\breve{d}_{i,k},\,i\in\mathcal{I}$, and computed from the LS estimate in~\eqref{E:mmWLLSEChanO} as follows
\begin{align}
\breve{d}_{i,k}\triangleq\mathop{\arg\max}_{d\in\mathcal{D}} |\hat{h}_{i,k}[d]|,\,k\in\mathcal{K},i\in\mathcal{I}.
\label{E:mmWLEstMGPDel}
\end{align}
In the above, $\hat{h}_{i,k}[d]$ is the $((k-1)(d_\mathrm{max}+1)+d+1)^\text{th}$ element of the vector LS estimate $\hat{\bm h}_i$ in~\eqref{E:mmWLLSEChanO}. Essentially, we choose the tap delay corresponding to the estimate with the largest magnitude as the estimated delay of the maximum gain path on each stream $i\in\mathcal{I}$. Then, for each stream $i\in\mathcal{I}$, the \emph{estimate} of the channel coefficient 
corresponding to the \emph{estimated} maximum gain path, is given by $\hat{h}_{i,k}[\breve{d}_{i,k}],\,i\in\mathcal{I},k\in\mathcal{K}$, which is the $((k-1)(d_\mathrm{max}+1)+\breve{d}_{i,k}+1)^\text{th}$ element of the vector LS estimate $\hat{\bm h}_i$ in~\eqref{E:mmWLLSEChanO}. 

Further, during the data transmission stage, the CU selects a set of streams $\mathcal{I}_k\subseteq\mathcal{I}$, for which the above estimated channel gains $\big|\hat{h}_{i,k}[\breve{d}_{i,k}]\big|$ corresponding to the estimated maximum gain paths, are larger than a given threshold; i.e., we define 
\begin{align}
\mathcal{I}_k\triangleq\begin{cases}\bigg\{i\in\mathcal{I}\bigg| \frac{P\big|\hat{h}_{i,k}[\breve{d}_{i,k}]\big|^2}{\sigma_{m_i}^2+\varepsilon_i^2}\geq\eta\bigg\},\\
\text{if }\exists\text{ at least one }i\in\mathcal{I}\text{ s.t. }\frac{P\big|\hat{h}_{i,k}\big[\breve{d}_{i,k}\big]\big|^2}{\sigma_{m_i}^2+\varepsilon_i^2}\geq\eta,\\
\{i_k\}\\\text{otherwise, where }i_k\triangleq\mathop{\arg\max}_{i\in\mathcal{I}}\big|\hat{h}_{i,k}\big[\breve{d}_{i,k}\big]\big|,\end{cases}\label{E:mmWLSelStream}
\end{align}
$k\in\mathcal{K}$, 
where $\eta$ is a suitable threshold. Also, let $\mathcal{I}_k=\{j_1,\dotsc,j_{I_k}\}$, and 
$|\mathcal{I}_k|=I_k$, so that $I_k\geq 1$, i.e. at least one stream is selected for decoding each user's signal. Thus, for the data transmission stage, we now have the \emph{reduced} set of 
streams $\check{y}_{\mathrm{d},i},\,i\in\mathcal{I} _k$ for each user $k$, with each $\check{y}_{\mathrm{d},i}$ given by~\eqref{E:mmWLDCSigCU}. On these set of streams, the CU performs path delay compensation as in~\cref{SS:mmWLPDCPerfCSILens}, but now with the estimated delays $\breve{d}_{i,k}$'s of the maximum gain paths instead of the true delays $d_{m_i,k,\ell^\star_{i,k}}$'s. The delay compensated signal can then be expressed in terms of the delay differences of the taps with $\breve{d}_{i,k}$, similar to~\eqref{E:mmWLDCSigCU}, as
\begin{align}
\breve{\check{y}}_{\mathrm{d},i,k}[n]&=\breve{h}_{i,kk}[\breve{0}]x_{\mathrm{d},k}[n]+\underbrace{\sum_{\breve{\nu}\in\Delta\setminus\{\breve{0}\}} \breve{h}_{i,kk}[\breve{\nu}]x_{\mathrm{d},k}[n-\breve{\nu}]}_{\text{ISI}}\notag\\
&\quad+\underbrace{\sum_{k'\in\mathcal{K}\setminus\{k\}}\sum_{\breve{\nu}\in\Delta}\breve{h}_{i,kk'}[\breve{\nu}]x_{\mathrm{d},k'}[n-\breve{\nu}]}_{\text{IUI}}\notag\\
&\quad+\breve{z}_{\mathrm{d},i,k}[n]+\breve{e}_{\mathrm{d},i,k}[n],\,i\in\mathcal{I}_k,
k\in\mathcal{K},\label{E:mmWLEstDCSigCULens}
\end{align}
where we have defined $\breve{\check{y}}_{\mathrm{d},i,k}[n]\triangleq\check{y}_{\mathrm{d},i,k}[n+\breve{d}_{i,k}]$, $\breve{z}_{\mathrm{d},i,k}[n]\triangleq z_{\mathrm{d},i}[n+\breve{d}_{i,k}]$, and $\breve{e}_{\mathrm{d},i,k}[n]\triangleq e_{\mathrm{d},i}[n+\breve{d}_{i,k}]$, and also placed a $\breve{}$ symbol over the delay differences $\breve{\nu}\in\Delta$ to indicate that they are w.r.t. the estimated tap delay $\breve{d}_{i,k}$. The channel coefficients $\breve{h}_{i,kk'}[\breve{\nu}]$ are defined similar to~\eqref{E:mmWLDCChanCoeff}, and are given by 
\begin{align}
\breve{h}_{i,kk'}[\breve{\nu}]&\triangleq\sum_{\ell'\in\mathcal{L}_{m_i,k'}}h_{i,k',\ell'}\delta\big[\breve{\nu}-\big(d_{m_i,k',\ell'}-\breve{d}_{i,k}\big)\big],\notag\\
&\quad i\in\mathcal{I}_k,\,k,k'\in\mathcal{K},\breve{\nu}\in\Delta.\label{E:mmWLEstDCChanCoeff}
\end{align}
Note that the $\breve{h}_{i,kk'}[\breve{\nu}]$'s in~\cref{E:mmWLEstDCSigCULens,E:mmWLEstDCChanCoeff} denote the \emph{true} channel coefficients that correspond to a delay difference of $\breve{\nu}$ with the \emph{estimated} delay $\breve{d}_{i,k}$. 
In vector form,~\eqref{E:mmWLEstDCSigCULens} can be written as 
\begin{align}
\breve{\check{\bm y}}_{\mathrm{d},k}[n]&=\breve{\bm h}_{kk}[\breve{0}]x_{\mathrm{d},k}[n]+\sum_{\breve{\nu}\in\Delta\setminus\{\breve{0}\}}\breve{\bm h}_{kk}[\breve{\nu}]x_{\mathrm{d},k}[n-\breve{\nu}]\notag\\
&\quad+\sum_{k'\in\mathcal{K}\setminus\{k\}}\sum_{\breve{\nu}\in\Delta}\breve{\bm h}_{kk'}[\breve{\nu}]x_{\mathrm{d},k'}[n-\breve{\nu}]\notag\\
&\quad+\breve{\bm z}_{\mathrm{d},k}[n]+\breve{\bm e}_{\mathrm{d},k}[n],\quad k\in\mathcal{K},\label{E:mmWLRxSigCULensV}
\end{align}
where all the vectors have dimensions $I_k\times 1$. Then, the \emph{reduced-size} linear MMSE beamformer $\breve{\bm u}_k\in\mathbb{C}^{I_k\times 1}$ for user $k$ after path delay compensation with the estimated tap delay of the maximum gain path on each stream $i\in\mathcal{I}_k$ is given by $\breve{\bm u}_k=\breve{\bm C}_k^{-1}\breve{\bm h}_{kk}[\breve{\nu}]$, where 
\begin{align}
\breve{\bm C}_k&\triangleq\sum_{\breve{\nu}\in\Delta\setminus\{\breve{0}\}}P\breve{\bm h}_{kk}[\breve{\nu}]\breve{\bm h}_{kk}[\breve{\nu}]^\mathsf{H}\notag\\&\quad+\sum_{k'\in\mathcal{K}\setminus\{k\}}\sum_{\breve{\nu}\in\Delta}P\breve{\bm h}_{kk'}[\breve{\nu}]\breve{\bm h}_{kk'}[\breve{\nu}]^\mathsf{H}+\breve{\bm\Sigma}+\breve{\bm \Xi},\label{E:mmWLRedCovMatEstDC}
\end{align}
is the covariance matrix of the interference terms in~\eqref{E:mmWLRxSigCULensV}, with 
$\breve{\bm\Sigma}\triangleq\mathrm{diag}\begin{psmallmatrix}\sigma_{m_{j_1}}^2&\dotsc&\sigma_{m_{j_{I_k}}}^2\end{psmallmatrix}$, and $\breve{\bm\Xi}\triangleq\begin{psmallmatrix}\varepsilon_{j_1}^2&\dotsc&\varepsilon_{j_{I_k}}^2\end{psmallmatrix}$. 

However, the CU cannot compute the beamformer $\breve{\bm u}_k$ since it only has the \emph{estimates} of $\breve{h}_{i,kk'}[\breve{\nu}]$, computed from~\eqref{E:mmWLLSEChanO}. Now, on the selected streams $i\in\mathcal{I}_k$, to avoid using the noisy estimates, the CU approximates the ISI and IUI terms in~\eqref{E:mmWLEstDCSigCULens} by thresholding the estimates  similar to~\eqref{E:mmWLSelStream}. With the thresholding, the estimates of $\breve{h}_{i,kk'}[\breve{\nu}]$, which we denote by $\hat{\breve{h}}_{i,kk'}[\breve{\nu}]$, are given by
\begin{align}
\hat{\breve{h}}_{i,kk'}[\breve{\nu}]&\triangleq\begin{cases}\hat{h}_{i,k'}[\breve{\nu}+\breve{d}_{i,k}],&\text{if }\frac{P|\hat{h}_{i,k'}[\breve{\nu}+\breve{d}_{i,k}]|^2}{\sigma_{m_i}^2+\varepsilon_i^2}\geq\eta,\\
0&\text{otherwise,}\end{cases}\label{E:mmWLEstDCEstChanCoeff}
\end{align}
$i\in\mathcal{I}_k,\,k,k'\in\mathcal{K},\breve{\nu}\in\Delta$. 
In~\eqref{E:mmWLEstDCEstChanCoeff}, $\hat{h}_{i,k'}[\breve{\nu}+\breve{d}_{i,k}]$ is the $((k'-1)(d_\mathrm{max}+1)+(\breve{\nu}+\breve{d}_{i,k})+1)^\text{th}$ element of the vector LS estimate $\hat{\bm h}_i$ in~\eqref{E:mmWLLSEChanO}. Note that due to the response of the lens antenna array, for each stream $i\in\mathcal{I}_k$, the thresholded estimates $\hat{\breve{h}}_{i,kk'}[\breve{\nu}]$ will be non-zero only for a few users $k'\in\mathcal{K}$ and delay differences $\breve{\nu}\in\Delta$. Let $\hat{\breve{\bm h}}_{kk'}[\breve{\nu}]=\begin{bsmallmatrix}\hat{\breve{h}}_{j_1,kk'}[\breve{\nu}]&\cdots &\hat{\breve{h}}_{j_{I_k},kk'}[\breve{\nu}]\end{bsmallmatrix}^\mathsf{T}\in\mathbb{C}^{I_k\times 1}$ denote the vector of the thresholded channel estimates in~\eqref{E:mmWLEstDCEstChanCoeff}. 
Then, the CU applies the \emph{reduced-size, approximate} linear MMSE beamforming $\hat{\breve{\bm u}}_k=\hat{\breve{\bm C}}_k^{-1}\hat{\breve{\bm h}}_{kk}[\breve{0}]$ on the observations $\breve{\check{\bm y}}_{\mathrm{d},k}[n]$ in~\eqref{E:mmWLRxSigCULensV}, where\footnote{It is assumed that the CU knows the quantization noise variances on each stream. These can be transmitted by the RRHs after the power probing stage to the CU once every frame period, along with information about the number of SQ bits on each antenna, since these quantities do not change
within a frame.}
\begin{align}
\hat{\breve{\bm C}}_k&\triangleq\sum_{\breve{\nu}\in\Delta\setminus\{\breve{0}\}}P\hat{\breve{\bm h}}_{kk}[\breve{\nu}]\hat{\breve{\bm h}}_{kk}[\breve{\nu}]^\mathsf{H}\notag\\
&\quad+\sum_{k'\in\mathcal{K}\setminus\{k\}}\sum_{\breve{\nu}\in\Delta}P\hat{\breve{\bm h}}_{kk'}[\breve{\nu}]\hat{\breve{\bm h}}_{kk'}[\breve{\nu}]^\mathsf{H}+\breve{\bm\Sigma}+\breve{\bm \Xi},\label{E:mmWLRedApproxEstCovMatEstDC}
\end{align} and the resulting SINR is given by 
\cref{E:mmWLSNRUkLEstCh} on the top of the next page. 
\begin{floatEq}
\begin{align}
\hat{\gamma}_k=\frac{P\big|\hat{\breve{\bm u}}^\mathsf{H}_k\breve{\bm h}_{kk}[\breve{0}]\big|^2}{\sum_{\breve{\nu}\in\Delta\setminus\{\breve{0}\}}P|\hat{\breve{\bm u}}^\mathsf{H}_k\breve{\bm h}_{kk}[\breve{\nu}]|^2+\sum_{k'\in\mathcal{K}\setminus k}\sum_{\breve{\nu}\in\Delta}P|\hat{\breve{\bm u}}^\mathsf{H}_k\breve{\bm h}_{kk'}[\breve{\nu}]|^2+\hat{\breve{\bm u}}^\mathsf{H}_k(\breve{\bm\Sigma}+\breve{\bm\Xi})\hat{\breve{\bm u}}_k}.\label{E:mmWLSNRUkLEstCh}
\end{align}
\end{floatEq}
Thus, with the proposed channel estimation, path delay compensation, and reduced-size, approximate linear MMSE receive beamforming for each user at the CU, a lower bound on the 
effective sum-throughput over all users of the CRAN is given by 
\begin{align}
\hat{r}_\mathrm{lens}=\bigg(1-\frac{T_\mathrm{a}+T_\mathrm{p}+2d_\mathrm{max}}{T_\mathrm{f}}\bigg)\sum_{k\in\mathcal{K}}\log_2(1+\hat{\gamma}_k)
\end{align}
in bps/Hz. 
In the following subsection we briefly describe a benchmark scheme in the conventional CRAN when the RRHs use UPAs and the users transmit using OFDM. 
\section{Benchmark: Conventional CRAN with UPAs and OFDM}\label{Sec:mmWLUPA}
In this case, each sector of an RRH is assumed to be equipped with a rectangular UPA of 
physical dimensions $\tilde{D}_y\times \tilde{D}_z$, 
with adjacent antenna elements of the array separated by a distance 
equal to half the wavelength. 
Let $\tilde{\mathcal{Q}}^{(j)}\triangleq\{1,\dotsc,\tilde{Q}\}$, denote the set of antennas at sector $j\in\mathcal{J}$ of each RRH in this case, and $\tilde{\mathcal{Q}}=\bigcup_{j\in\mathcal{J}}\mathcal{Q}^{(j)}$ denote the set of all antenna elements. The number of antenna elements per sector is given by $\tilde{Q}=\tilde{Q}_y\tilde{Q}_z$, where $\tilde{Q}_y=\lfloor 2\tilde{D}_y\rfloor$ denotes the number of antennas along the $y$-axis of the UPA and $\tilde{Q}_z=\lfloor 2\tilde{D}_z\rfloor$ is the number of antennas along the $z$-axis. 
For the UPA, the array response vector $\bm a(\theta,\phi)\in\mathbb{C}^{\tilde{Q}\times 1}$ in~\eqref{E:mmWLChanTap} 
can be written as 
\begin{align}
\bm a(\theta,\phi)=\bm a_{z}(\theta)\otimes\bm a_{y}(\theta,\phi),\label{E:mmWLARUPA}
\end{align} 
where $\bm a_z(\theta)=\sqrt{\tfrac{\tilde{D}_z}{\tilde{Q}_z}}\begin{bsmallmatrix}
1& \exp(\jmath\pi\sin\theta)&\cdots & \exp(\jmath\pi(\tilde{Q}_z-1)\sin\theta)
\end{bsmallmatrix}^\mathsf{T}\in\mathbb{C}^{\tilde{Q}_z\times 1}$ is the response of the linear array along the $z$-axis, and $\bm a_y(\theta,\phi)=\sqrt{\tfrac{\tilde{D}_y}{\tilde{Q}_y}}\begin{bsmallmatrix}
1& \exp(\jmath\pi\cos\theta\sin\phi)&\cdots & \exp(\jmath\pi(\tilde{Q}_y-1)\cos\theta\sin\phi)
\end{bsmallmatrix}^\mathsf{T}\in\mathbb{C}^{\tilde{Q}_y\times 1}$ is the response of the linear array along the $y$-axis. The total channel bandwidth $W$ is divided into $N$ orthogonal sub-channels~(SCs) of equal width denoted by $\mathcal{N}\triangleq\{0,1,\dotsc,N-1\}$. 

The uplink transmission protocol is similar to that in~\cref{F:mmWLFrameStruct}, except that the users transmit their symbols in $(N+\mu)$-length symbol blocks using OFDM, 
where $\mu\geq d_\mathrm{max}$ is the length of the cyclic prefix~(CP). Hence, the guard intervals shown in~\cref{F:mmWLFrameStruct} are implicit in this case due to the CP, which is discarded at the CU. We assume that all the users transmit on all SCs for fully exploiting spatial multiplexing. 
In general, there can be multiple OFDM symbols in the power probing, channel estimation, and data transmission stages. For convenience, it is assumed that $T_\mathrm{a}$, $T_\mathrm{p}$ and $T_\mathrm{d}$ are all integer multiples of $N+\mu$, i.e. there are $\tau_\mathrm{a}=\tfrac{T_\mathrm{a}}{N+\mu}\in\mathbb{Z}_{++}$, $\tau_\mathrm{p}=\frac{T_\mathrm{p}}{N+\mu}\in\mathbb{Z}_{++}$, and $\tau_\mathrm{d}=\tau_\mathrm{d}\in\mathbb{Z}_{++}$ OFDM symbols in the power probing, channel estimation, and data transmission stages of a frame, respectively. 
Let $\tilde{\bm x}_{t,k}\in\mathbb{C}^{N\times 1}$ denote the $t^\text{th}$ OFDM symbol transmitted by user $k$ in the frequency-domain. Then the corresponding time-domain signals are 
given by the inverse discrete Fourier transform~(IDFT) $\bm x_{t,k}=\bm F^\mathsf{H}\tilde{\bm x}_k$, where $\bm F\in\mathbb{C}^{N\times N}$ denotes the DFT matrix with 
columns $\bm f_n\triangleq\tfrac{1}{\sqrt{N}}\begin{bsmallmatrix}1&\exp(-\jmath\frac{2\pi n}{N})&\cdots &\exp(-\jmath\frac{2\pi n(N-1)}{N})\end{bsmallmatrix}^\mathsf{T},\,n\in\mathcal{N}$. Similar to~\cref{SS:mmWLRRHProc}, the RRHs perform antenna selection and SQ bit allocation on the received time-domain signals and forward them to the CU. Let $\tilde{\mathcal{Q}}_m\triangleq\{q\in\tilde{\mathcal{Q}}|b_{q,m}>0\}$ with $|\tilde{\mathcal{Q}}_m|=\tilde{Q}_m$ denote the set of selected antennas at each RRH $m$ in this case. Notice that with UPAs, the selected set of antennas and their non-zero quantization bit allocation at each RRH can be very different from that with the lens antenna array, due to the lack of the lens focusing. 
Let $\tilde{Q}_\mathrm{tot}=\sum_{m\in\mathcal{M}}\tilde{Q}_m$ and similar to~\cref{SS:mmWLRRHProc} we index the selected antennas~(streams) with $i\in\{1,\dotsc,\tilde{Q}_\mathrm{tot}\}\triangleq\tilde{\mathcal{I}}$, where stream index $i$ corresponds to 
antenna $q_i
$ of RRH $m_i
$. 
\subsection{Channel Estimation}\label{SS:mmWLCEUPA}
In the channel estimation stage, let $\tilde{x}_{\mathrm{p},t,k}[n]$ 
denote the frequency-domain pilot symbol transmitted by user $k\in\mathcal{K}$ on SC $n\in\mathcal{N}$ in the $t^\text{th}$ OFDM symbol, where $t\in\{0,1,\dotsc,\tau_\mathrm{p}-1\}$. Also, let $\tilde{\bm X}_{\mathrm{p},t,k}\triangleq\mathrm{diag}\begin{psmallmatrix}\tilde{x}_{\mathrm{p},t,k}[0]&\cdots &\tilde{x}_{\mathrm{p},t,k}[N-1]\end{psmallmatrix}\in\mathbb{C}^{N\times N}$ be the diagonal matrix constructed from the pilot symbols on all the SCs, and $\tilde{\bm h}_{i,k}=\sqrt{N}\bm F_0\bm h_{i,k}$, the vector of frequency-domain channel coefficients from user $k$ to antenna $i\in\mathcal{I}$, 
where $\bm h_{i,k}
\in\mathbb{C}^{(d_\mathrm{max}+1)\times 1}$ is the time-domain channel coefficient vector similar to that in~\cref{SS:mmWLImperfCSILens}, and 
$\bm F_0\in\mathbb{C}^{N\times (d_\mathrm{max}+1)}$ denotes the first $(d_\mathrm{max}+1)$ columns of $\bm F$. Then, the frequency-domain received signal vector at the CU corresponding to the OFDM symbol $t$ at antenna $i\in\tilde{\mathcal{I}}$, after removing the CP and applying the DFT can be expressed as 
\begin{align}
\tilde{\check{\bm y}}_{\mathrm{p},t,i}&=\sum_{k\in\mathcal{K}}\tilde{\bm X}_{\mathrm{p},t,k}\tilde{\bm h}_{i,k}+\tilde{\bm z}_{\mathrm{p},t,i}+\tilde{\bm e}_{\mathrm{p},t,i}\notag\\
&=\sum_{k\in\mathcal{K}}\sqrt{N}\tilde{\bm X}_{\mathrm{p},t,k}\bm F_0\bm h_{i,k}+\tilde{\bm z}_{\mathrm{p},t,i}+\tilde{\bm e}_{\mathrm{p},t,i}\notag\\
&=\bar{\bm X}_{\mathrm{p},t}\bm h_i+\tilde{\bm z}_{\mathrm{p},t,i}+\tilde{\bm e}_{\mathrm{p},t,i},\label{E:mmWLFDRxSigQCUV}
\end{align}
where $\tilde{\bm z}_{\mathrm{p},t,i}=\bm F\bm z_{\mathrm{p},t,i}$ is the DFT of the AWGN so that $\tilde{\bm z}_{\mathrm{p},t,i}\sim\mathcal{CN}(\bm 0,\sigma_{m_i}^2\bm I_N)$, 
and $\tilde{\bm e}_{\mathrm{p},t,i}=\bm F\bm e_{\mathrm{p},t,i}$ is the DFT of the quantization noise with $\mathbb{E}[\tilde{\bm e}_{\mathrm{p},t,i}]=\bm 0$ and $\mathbb{E}[\tilde{\bm e}_{\mathrm{p},t,i}\tilde{\bm e}_{\mathrm{p},t,i}^\mathsf{H}]=\varepsilon_i^2\bm I_N$. Since each element of $\tilde{\bm e}_{\mathrm{p},t,i}$ is a linear combination of the $N$ elements in $\bm e_{\mathrm{p},t,i}$, due to the central limit theorem, $\tilde{\bm e}_{\mathrm{p},t,i}\sim\mathcal{CN}(\bm 0,\varepsilon_i^2\bm I_N)$. Also, in~\eqref{E:mmWLFDRxSigQCUV}, $\bar{\bm X}_{\mathrm{p},t}\triangleq\sqrt{N}\begin{bsmallmatrix}\tilde{\bm X}_{\mathrm{p},t,1}\bm F_0&\cdots &\tilde{\bm X}_{\mathrm{p},t,K}\bm F_0\end{bsmallmatrix}\in\mathbb{C}^{N\times K(d_\mathrm{max}+1)}$, and $\bm h_i$ is defined similar to~\eqref{E:mmWLRxSigQCUCELensV}. 
For $\tau_\mathrm{p}>1$, stacking each of 
$\tilde{\check{\bm y}}_{\mathrm{p},t,i}$, 
$\bar{\bm X}_{\mathrm{p},t}$, $\tilde{\bm z}_{\mathrm{p},t,i}$, and $\tilde{\bm e}_{\mathrm{p},t,i},\,t\in\{1,\dotsc,\tau_\mathrm{p}-1\}$ 
column-wise, 
we have 
\begin{align}
\tilde{\check{\bm y}}_{\mathrm{p},i}&=\bar{\bm X}_\mathrm{p}\bm h_i+\tilde{\bm z}_{\mathrm{p},i}+\tilde{\bm e}_{\mathrm{p},i}.\label{E:mmWLFDRxSigQCUAllSym}
\end{align}
Similar to~\cref{SS:mmWLImperfCSILens}, the LS estimate for $\bm h_i$ is given by\footnote{Since both $\tilde{\bm z}_{\mathrm{p},i}$ and $\tilde{\bm e}_{\mathrm{p},i}$ are CSCG, this is also the maximum likelihood estimate of $\bm h_i$.}
\begin{align}
\hat{\bm h}_i=\bar{\bm X}_\mathrm{p}^\dag\tilde{\check{\bm y}}_\mathrm{p}=\bm h_i+\bm\xi_{\mathrm{p},i},
\label{E:mmWLLSChanEstUPA}
\end{align}
where $\bm\xi_{\mathrm{p},i}\triangleq \bar{\bm X}_\mathrm{p}^\dag(\tilde{\bm z}_{\mathrm{p},i}+\tilde{\bm e}_{\mathrm{p},i})\sim\mathcal{CN}(\bm 0,(\sigma_{m_i}^2+\varepsilon_i^2)(\bar{\bm X}_\mathrm{p}^\mathsf{H}\bar{\bm X}_\mathrm{p})^{-1})$. Note that the estimate in~\eqref{E:mmWLLSChanEstUPA} exists only if $\tau_\mathrm{p}N\geq K(d_\mathrm{max}+1)$, and 
the 
MSE of the estimate is 
\begin{align}
\mathbb{E}[\|\hat{\bm h}_i-\bm h_i\|^2]=\mathbb{E}[\|\bm\xi_{\mathrm{p},i}\|^2]=(\sigma_{m_i}^2+\varepsilon_i^2)\mathrm{tr}\big((\bar{\bm X}_\mathrm{p}^\mathsf{H}\bar{\bm X}_\mathrm{p})^{-1}\big),
\end{align}
which is minimized if $\bar{\bm X}_\mathrm{p}^\mathsf{H}\bar{\bm X}_\mathrm{p}=\tilde{c}\bm I_{K(d_\mathrm{max}+1)}$, where $\tilde{c}$ is a constant. Due to the construction of $\bar{\bm X}_\mathrm{p}$ in~\eqref{E:mmWLFDRxSigQCUAllSym}, this translates to the requirement that 
\begin{align}
\sum_{t=0}^{\tau_\mathrm{p}-1}\bar{\bm X}_{\mathrm{p},t,k}^\mathsf{H}\bar{\bm X}_{\mathrm{p},t,k'}
&=N\bm F_0^\mathsf{H}\bigg(\sum_{t=0}^{\tau_\mathrm{p}-1}\tilde{\bm X}_{\mathrm{p},t,k}^\mathsf{H}\tilde{\bm X}_{\mathrm{p},t,k'}\bigg)\bm F_0\notag\\
&=\delta[k-k']\tilde{c}\bm I_{(d_\mathrm{max}+1)}\,k,k'\in\mathcal{K}.\label{E:mmWLNecCondTSUPA}
\end{align}
Now, let $\bm\Psi\triangleq\bm V\otimes\bm S_{\mathrm{p},k}\in\mathbb{C}^{\tau_\mathrm{p}N\times\tau_\mathrm{p}N}$, where $\bm V\in\mathbb{C}^{\tau_\mathrm{p}\times\tau_\mathrm{p}}$ is a unitary matrix and $\bm S_{\mathrm{p},k}\triangleq\mathrm{diag}\begin{psmallmatrix}s_{\mathrm{p},k}[0]&\cdots&s_{\mathrm{p},k}[N-1]\end{psmallmatrix}\in\mathbb{C}^{N\times N}$ is a diagonal matrix with unit amplitude training symbols, satisfying $\bm S_{\mathrm{p},k}^\mathsf{H}\bm S_{\mathrm{p},k}=\bm I_N$. Then, $\bm\Psi^\mathsf{H}\bm\Psi=\bm V^\mathsf{H}\bm V\otimes\bm S_{\mathrm{p},k}^\mathsf{H}\bm S_{\mathrm{p},k}=\bm I_{\tau_\mathrm{p}N}$. The 
condition~\eqref{E:mmWLNecCondTSUPA} can then be ensured by choosing $\tilde{\bm X}_{\mathrm{p},t,k}=\sqrt{P}\bm\Psi_{t\omega}(\sqrt{N}\mathrm{diag}(\bm f_\kappa))$, where $\bm\Psi_{t\omega}\in\mathbb{C}^{N\times N}$ is the $(t,\omega)^\text{th}$ block of $\bm\Psi$, $t,\omega\in\{0,\dotsc,\tau_\mathrm{p}-1\}$~\cite[Theorem 2]{chi-etal2011training}. Here, $\kappa$ and $\omega$ are such that the user index $k$ can be written as $k=\kappa\tau_\mathrm{p}+\omega+1$, i.e., $\kappa=\lfloor\tfrac{k-1}{\tau_\mathrm{p}}\rfloor\in\{0,1,\dotsc,\lfloor\tfrac{K-1}{\tau_\mathrm{p}}\rfloor\}$, and $\omega=k-1-\kappa\tau_\mathrm{p}\in\{0,1,\dotsc,\tau_\mathrm{p}-1\}$. Then,~\eqref{E:mmWLNecCondTSUPA} 
holds with $\tilde{c}=\tau_\mathrm{p}NP$, and $\bar{\bm X}_\mathrm{p}^\mathsf{H}\bar{\bm X}_\mathrm{p}=\tau_\mathrm{p}NP\bm I_{K(d_\mathrm{max}+1)}$. Then, the LS estimate in~\eqref{E:mmWLLSChanEstUPA} can be written as
\begin{align}
\hat{\bm h}_i=\frac{1}{\tau_\mathrm{p}NP}\bar{\bm X}_\mathrm{p}^\mathsf{H}\tilde{\check{\bm y}}_\mathrm{p}=\bm h_i+\bm\xi_{\mathrm{p},i},
\end{align}
where $\bm\xi_{\mathrm{p},i}=\tfrac{1}{\tau_\mathrm{p}NP}\bar{\bm X}_\mathrm{p}^\mathsf{H}(\tilde{\bm z}_{\mathrm{p},i}+\tilde{\bm e}_{\mathrm{p},i})\sim\mathcal{CN}\Big(\bm 0,\tfrac{\sigma_{m_i}^2+\varepsilon_i^2}{\tau_\mathrm{p}NP}\bm I_{K(d_\mathrm{max}+1)}\Big)$. 
Note that
~\eqref{E:mmWLLSChanEstUPA} gives the time-domain channel estimates, and 
the frequency domain 
estimates $\hat{\tilde{\bm h}}_{i,k},\,i\in\tilde{\mathcal{I}}$ for each user $k\in\mathcal{K}$, are given by applying the transformation $\hat{\tilde{\bm h}}_{i,k}=\sqrt{N}\bm F_0\hat{\bm h}_{i,k},\,i\in\tilde{\mathcal{I}},k\in\mathcal{K}$, where $\hat{\bm h}_{i,k}$ is the length $(d_\mathrm{max}+1)$ sub-vector in $\hat{\bm h}_i$ starting at index $(k-1)(d_\mathrm{max}+1)+1$ and ending at $k(d_\mathrm{max}+1)$. 
\subsection{Achievable Sum-Rate}
Let $\tilde{x}_{\mathrm{d},k}[n]=\sqrt{P}s_{\mathrm{d},k}[n]$ denote the frequency-domain data symbol transmitted by user $k$ on SC $n\in\mathcal{N}$, where $s_{\mathrm{d},k}[n]\sim\mathcal{CN}(0,1)$ and $P\geq 0$. 
Then, the frequency-domain received signal at the CU on SC $n\in\mathcal{N}$ 
corresponding to stream $i\in\tilde{\mathcal{I}}$ 
can be expressed as
\begin{align}
\tilde{\check{y}}_{\mathrm{d},i}[n]&=\sum_{k\in\mathcal{K}}\tilde{h}_{i,k}[n]\tilde{x}_{\mathrm{d},k}[n]+\tilde{z}_{\mathrm{d},i}[n]+\tilde{e}_{\mathrm{d},i}[n],
\end{align} 
where $\tilde{z}_{\mathrm{d},i}[n]$ and $\tilde{e}_{\mathrm{d},i}[n]$ are the $n^\text{th}$ components of the DFTs of the AWGN and quantization noise similar to that in~\eqref{E:mmWLFDRxSigQCUV}. Collecting the received signals from all streams on each SC $n$, we have
\begin{align}
\tilde{\check{\bm y}}_\mathrm{d}[n]&=\sum_{k\in\mathcal{K}}\tilde{\bm h}_k[n]\tilde{x}_{\mathrm{d},k}[n]+\tilde{\bm z}_\mathrm{d}[n]+\tilde{\bm e}_\mathrm{d}[n],n\in\mathcal{N},\label{E:mmWLRxVCUSC}
\end{align}
where all the vectors are of dimensions $\tilde{Q}_\mathrm{tot}\times 1$. If the CU has perfect CSI, 
the 
linear 
MMSE beamforming vector $\tilde{\bm u}_{k,n}\in\mathbb{C}^{\tilde{Q}_{\mathrm{tot}}\times 1}$ that maximizes the receive SINR for each user $k$ on SC $n$, treating the interference from other users as noise, is given by 
$\tilde{\bm u}_{k,n}=\tilde{\bm C}_{k,n}^{-1}\tilde{\bm h}_{k}[n]$, where 
\begin{align}
\tilde{\bm C}_{k,n}\triangleq\sum_{k'\in\mathcal{K}\setminus\{k\}}P\tilde{\bm h}_{k'}[n]\tilde{\bm h}_{k'}[n]^\mathsf{H}+\tilde{\bm\Sigma}
+\tilde{\bm \Xi},
\end{align}
with $\tilde{\bm\Sigma}\triangleq\mathrm{blkdiag}\begin{psmallmatrix}
\sigma_1^2\bm I_{\tilde{Q}_{1}}&\cdots&\sigma_M^2\bm I_{\tilde{Q}_M}
\end{psmallmatrix}\in\mathbb{C}^{\tilde{Q}_\mathrm{tot}\times \tilde{Q}_\mathrm{tot}}$ and 
$\tilde{\bm \Xi}\triangleq\mathrm{diag}\begin{psmallmatrix}\{\varepsilon_i^2\}_{i=1}^{Q_\mathrm{tot}}\end{psmallmatrix}\in\mathbb{C}^{\tilde{Q}_\mathrm{tot}\times \tilde{Q}_\mathrm{tot}}$. The SINR for user $k$ on SC $n$ is then given by $\tilde{\gamma}_{k,n}=P\tilde{\bm h}_{k}[n]^\mathsf{H}\tilde{\bm C}_{k,n}^{-1}\tilde{\bm h}_{k}[n]$, and thus the sum-rate over all users with UPAs and OFDM transmission is given by
\begin{align}
\tilde{r}_\mathrm{upa}=\sum_{k\in\mathcal{K}}\frac{1}{N+\mu}\sum_{n\in\mathcal{N}}\log_2(1+\tilde{\gamma}_{k,n})
\end{align}
in bps/Hz.

On the other hand, with estimated CSI, the CU computes the linear MMSE receiver $\hat{\tilde{\bm u}}_{k,n}=\hat{\tilde{\bm C}}_{k,n}^{-1}\hat{\tilde{\bm h}}_k[n]$ on each SC $n$, where the matrix $\hat{\tilde{\bm C}}_{k,n}\triangleq\sum_{k'\in\mathcal{K}\setminus\{k\}}P\hat{\tilde{\bm h}}_{k'}[n]\hat{\tilde{\bm h}}_{k'}[n]^\mathsf{H}+\tilde{\bm\Sigma}+\tilde{\bm \Xi}$ and the components of $\hat{\tilde{\bm h}}_k[n],\,n\in\mathcal{N},k\in\mathcal{K}$ are computed by transforming the time-domain estimates in~\eqref{E:mmWLLSChanEstUPA} as outlined earlier. Then the 
resulting SINR is given by 
\begin{align}
\hat{\tilde{\gamma}}_{k,n}=\frac{P|\hat{\tilde{\bm u}}^\mathsf{H}_{k,n}\tilde{\bm h}_{k}[n]|^2}{\sum_{k'\in\mathcal{K}\setminus\{k\}}P|\hat{\tilde{\bm u}}_{k,n}^\mathsf{H}\tilde{\bm h}_{k'}[n]|^2+\hat{\tilde{\bm u}}_{k,n}^\mathsf{H}(\tilde{\bm\Sigma}+\tilde{\bm \Xi})\hat{\tilde{\bm u}}_{k,n}}
\end{align}
Thus, the effective sum-throughput over all users 
with UPAs and OFDM transmission 
is given by 
\begin{align}
\hat{\tilde{r}}_\mathrm{upa}=\Big(1-\frac{\tau_\mathrm{a}+\tau_\mathrm{p}}{\tau_\mathrm{f}}\Big)\frac{1}{N+\mu}\sum_{k\in\mathcal{K}}\sum_{n\in\mathcal{N}}\log_2(1+\hat{\tilde{\gamma}}_{k,n}),
\end{align}
in bps/Hz. 
\section{Simulation Results}\label{Sec:mmWLSimRes}
For the simulations, we consider a part of the general sectorized CRAN illustrated in~\cref{F:mmWLSysMod} with $M=6$ RRHs located at the corners of a hexagon of side $50/\sqrt{3}$~meters~(m), along with $K=6$ users located randomly and uniformly in the common 
region covered by one sector of each RRH as shown in~\cref{F:mmWLSimLayout}. All the RRHs are assumed to be of height $30$~meters~(m), while the users' height can vary between $1$~m and $25$~m. 
\begin{figure}[t]
\centering
\includegraphics[width=\columnwidth]{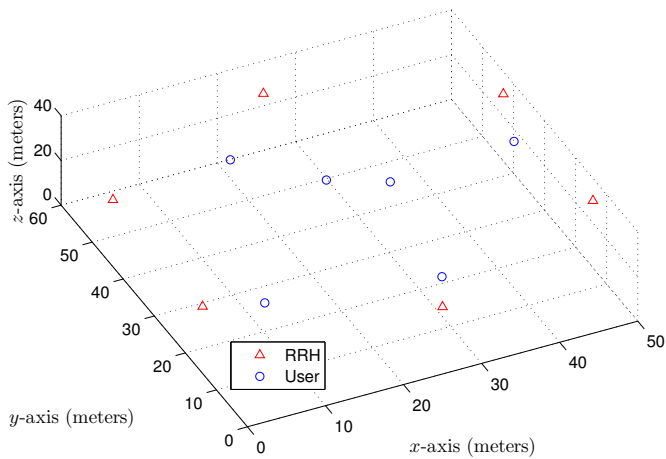}
\caption{Example of RRH and user layout for simulations.}\label{F:mmWLSimLayout}
\end{figure}
The size of the EM lens and the UPA at each RRH is chosen such that $D_y=D_z=10=\tilde{D}_y=\tilde{D}_z$, so that the aperture size is the same in both cases. The height of the RRHs is greater than the maximum height of the users, and the minimum distance along the ground between a user and RRH is assumed to be $2$~m. Based on the above setup, the maximum coverage angles for the lens array in the elevation direction are chosen to be $\Theta_+=\pi/2,\Theta_-=\pi/6$, while the corresponding angles in the azimuth direction are chosen to be $\Phi_+=\Phi_-=\pi/3$ due to the sectorization. This leads to $Q=208$ and $\tilde{Q}=400$ antennas per sector at each RRH for the lens array and the UPA, respectively. The carrier frequency is $28$~GHz and the mmWave channel shared by the users and RRHs has a bandwidth $W=200$~MHz. The maximum delay spread of the channel is $100$~ns, which translates to $d_\mathrm{max}=20$ symbols. The free-space path loss between a user $k$ and RRH $m$ separated by distance $D'_{m,k}$~m is modeled by $61.4+34.1\log_{10}D'_{m,k}$~dB~\cite{samimi-rappaport2014ultra}. 
The number of paths $L_{m,k}$ between user $k$ and RRH $m$ is equal to $1$, $2$ or $3$ with equal probability, and the time delays of the paths $\zeta_\ell,\,\ell\in\mathcal{L}_{m,k}$, are generated from an exponential distribution with mean $r_\zeta\mu_\zeta$, but truncated to $100$~ns, where $\mu_\zeta=67$~ns and $r_\zeta=0.25$~\cite{samimi-rappaport2014ultra}. The relative power levels of the paths are given by 
$\bar{\kappa}_{m,k,\ell}=\kappa_{m,k,\ell}/(\sum_{\ell'\in\mathcal{L}_{m,k}}\kappa_{m,k,\ell'})$, where $\kappa_{m,k,\ell}=10^{\frac{Z}{10}}(0.613)\exp\big(\tfrac{-\zeta_\ell}{31.4~\text{ns}}\big)$~\cite{samimi-rappaport2014ultra}, and $Z$ is zero-mean Gaussian with standard deviation $9.4$~dB, representing the per-path shadowing~\cite{samimi-rappaport2014ultra}; thus $|\alpha_{m,k,\ell}|^2=\bar{\kappa}_{m,k,\ell}P_{\mathrm{r},m,k}$, where $P_{\mathrm{r},m,k}$ in watt is the received power at RRH $m$ from user $k$ after accounting for the free-space path loss.  The phase shifts of each path $\angle\alpha_{m,k,\ell}$ are generated uniformly randomly from the interval $[0,2\pi]$. The elevation angles of arrival $\theta_{m,k,\ell}$ for each path $\ell\in\mathcal{L}_{m,k}$ from user $k\in\mathcal{K}$ to RRH $m\in\mathcal{M}$ are independent and generated as $\theta_{m,k,\ell}=\bar{\theta}_{m,k,\ell}+\mathcal{U}[-\tfrac{\pi}{12},\tfrac{\pi}{12}]$, 
where $\bar{\theta}_{m,k,\ell}$ is the line-of-sight~(LoS) angle of elevation from the user to the RRH. Similarly, the azimuth angles of arrival $\phi_{m,k,\ell}$ are independent and generated 
as $\phi_{m,k,\ell}=\bar{\phi}_{m,k,\ell}+\mathcal{U}[-\tfrac{\pi}{6},\tfrac{\pi}{6}]$, where $\bar{\phi}_{m,k,\ell}$ is the LoS angle of azimuth. The power spectral density of the AWGN is $-174$~dBm/Hz, with an additional noise figure of $6$~dB at each RRH. The inter-sector/inter-cluster interference is assumed to be equal to $-80$~dBm at all sectors. The fronthaul capacity of all RRHs are the same, i.e. $\bar{R}_m=\bar{R},\forall m\in\mathcal{M}$. 
For the benchmark scheme using OFDM with UPAs at the RRHs, the number of SCs $N=256$ and the CP length $\mu=d_\mathrm{max}=20$. The minimum coherence time of the channels is taken to be $0.4$~ms, 
while each frame is assumed to be $T_\mathrm{f}=8280$ symbols long, which corresponds to $30$ OFDM symbols. For the proposed system with lens arrays, we choose $T_\mathrm{p}=K(d_\mathrm{max}+1)$, while for the benchmark system with UPAs and OFDM, $T_\mathrm{p}=3(N+\mu)$ i.e., $\tau_\mathrm{p}=3$. In order to make a fair comparison, we assume $T_\mathrm{a}=N+d_\mathrm{max}$ in both cases. All the results are averaged over 
random user locations and 
channel realizations. 
\begin{figure}[t]
\centering
\includegraphics[width=\columnwidth]{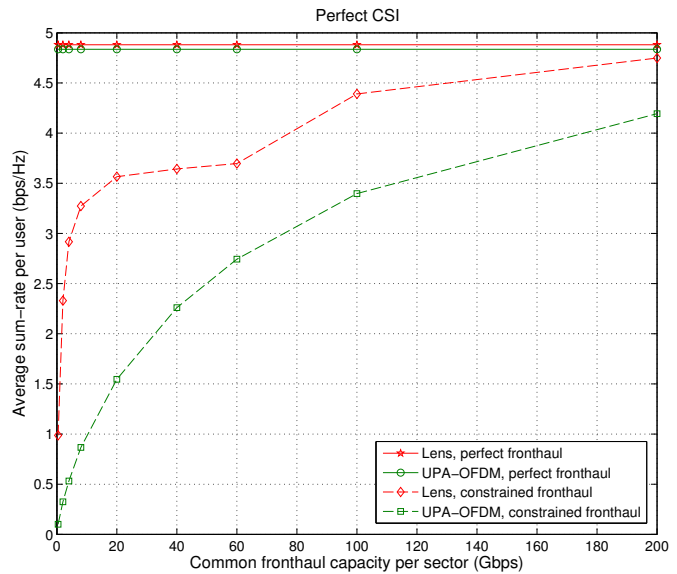}%
\caption{Average rate per user vs.\ common fronthaul capacity per RRH sector under perfect CSI, with $P=23$~dBm.}
\label{F:mmWLSEvsFHRPerfCSI}
\end{figure}

\cref{F:mmWLSEvsFHRPerfCSI,F:mmWLSEvsFHREstCSI} 
plot the average throughput per user against the common fronthaul capacity per sector $\bar{R}/3$, where the users' transmit power $P=23$~dBm. 
\cref{F:mmWLSEvsFHRPerfCSI} shows the performance assuming perfect CSI at the CU, both with and without the fronthaul constraints at the RRHs. Without fronthaul constraints and with perfect CSI, it can be observed that the performance of the benchmark system with OFDM and UPAs is almost similar to the proposed system with lens. The small difference can be attributed to the loss due to the CP insertion in the case of OFDM transmission with UPAs. As the lens array only redistributes the incident energy among the antenna elements depending on the angle of arrival of the plane waves, it is reasonable to expect that it does not offer a gain in the achievable rate over UPA-OFDM when there is no constraint on the fronthaul. However, since only single-carrier transmission by the users is required, and the ISI and IUI can be mitigated without CP insertion and IDFT/DFT processing, the use of lens arrays can greatly simplify the signal processing required, when compared with UPA-OFDM. For the system with UPA-OFDM, the overall complexity of performing linear MMSE beamforming for all the users on all the SCs with no fronthaul constraint is of the order $O(NK(MJ\tilde{Q})^3)$, while for the proposed system with lens arrays, the complexity is $O(K(MJQ)^3)$, offering a reduction of a factor of $N$. 

However, when the fronthaul is constrained, it can be observed from~\cref{F:mmWLSEvsFHRPerfCSI} that the proposed system with lens arrays 
provides significant throughput gains over the conventional UPA-OFDM. This is due to the energy focusing property of the EM lens, which makes the antenna selection and quantization bit allocation much more effective with lens arrays compared to UPAs. In the antenna selection and bit allocation scheme, the antennas with larger received power are selected and allocated 
more bits compared to the antennas with lower received power. 
\cref{T:mmWLNSelAntpRRH} compares the number of antennas selected at the RRHs in both cases. 
\begin{table}[t]
\renewcommand{\arraystretch}{1.3}
\caption{Average number of selected antennas per RRH}
\label{T:mmWLNSelAntpRRH}
\centering
\begin{tabular}{|p{1.8cm}|c|c|c|c|c|c|c|c|c|} \hline Fronthaul capacity per sector~(Gbps) & 0.4 & 2 & 4 & 8 & 20 & 40 & 60 & 100 & 200 \\ \hline Lens & 2 & 5 & 10 & 19 & 48 & 98 & 148 & 208 & 208\\ UPA-OFDM & 2 & 6 & 11 & 21 & 51 & 101 & 151 & 251 & 400\\ \hline
\end{tabular}
\end{table} 
As the EM lens focuses the incident energy onto a few antenna elements, the selection and bit allocation algorithm tends to select these elements with large received power. Thus, although the number of antennas selected in both cases is more or less similar when the fronthaul capacity is low, the selection and bit allocation is much more effective with lens arrays when compared to UPA-OFDM. With UPA-OFDM, the received energy distribution is more or less similar across the antenna elements and 
the selection and bit allocation can 
thus adversely affect the achievable rate compared to the case with no fronthaul constraint. 
\begin{figure}[t]
\centering
\includegraphics[width=\columnwidth]{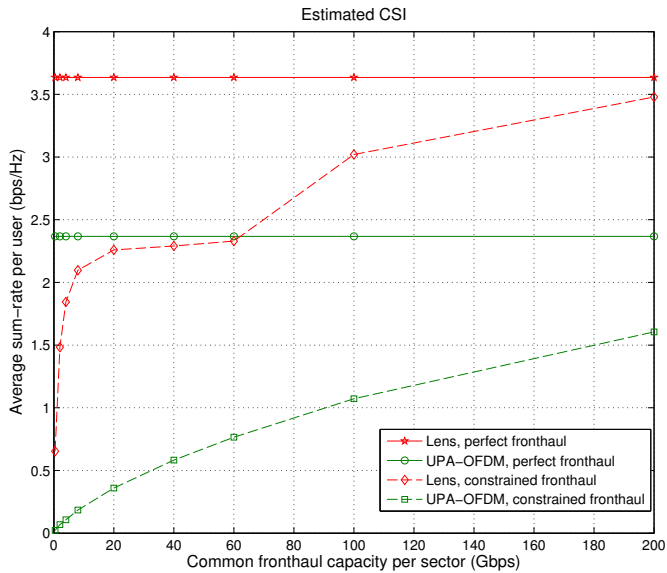}%
\caption{Average rate per user vs.\ common fronthaul capacity per RRH sector under estimated CSI, with $P=23$~dBm.}
\label{F:mmWLSEvsFHREstCSI}
\end{figure}
\cref{F:mmWLSEvsFHREstCSI} shows the performance of the proposed system with estimated CSI. For the thresholding of the channel estimates as given in~\cref{E:mmWLSelStream,E:mmWLEstDCEstChanCoeff}, we choose $\eta=3$~dB. Again, it can be observed that with the proposed channel estimation scheme via thresholding the LS channel estimates, the system with lens antenna arrays offers significant gains over the benchmark UPA-OFDM system. This is because the quantization and antenna selection at the RRHs adversely affects both the channel estimation as well as data transmission in the benchmark system, when the fronthaul is constrained. On the other hand, due to the energy focusing property of the lens, the stronger channels are relatively less affected by the quantization, and since the channels between the users and RRHs are sparse in the angular domain, and the lens converts this angular domain sparsity to the spatial domain, it is sufficient to use only the estimates of the stronger channels to retain most of the gains in the subsequent MMSE beamforming. 
In the UPA-OFDM system, such a thresholding is not effective since the energy is distributed uniformly over the antenna elements, and there is no way of distinguishing the significant channels in the spatial domain. \cref{T:mmWLNStreamspUsr} shows the average number of streams $I_k$ selected after thresholding the channel estimates for the proposed reduced, approximate MMSE beamforming with lens arrays at the RRHs. It can be observed that the total number of streams can potentially be even an order of magnitude less than the total number of streams in the UPA-OFDM case~(see \cref{T:mmWLNSelAntpRRH}), and this leads to a huge reduction in complexity of the receive beamforming in CRAN, since for the proposed system the complexity is only $O(KI_\mathrm{max}^3)$, where $I_\mathrm{max}=\max_{k\in\mathcal{K}}I_k$ is the maximum number of streams for a user, while for the benchmark with UPA-OFDM, it is $O(NK\tilde{Q}_\mathrm{tot}^3)$. Moreover, this gain in throughput is achieved at a much lower training overhead of $K(d_\mathrm{max}+1)$ for the proposed system, compared to at least $N+d_\mathrm{max}$ for the benchmark, where in practice $N\gg K$ as usual. 
\begin{table}[t]
\renewcommand{\arraystretch}{1.3}
\caption{Average number of streams selected per user for reduced MMSE with estimated CSI for lens antenna array}
\label{T:mmWLNStreamspUsr}
\centering
\begin{tabular}{|p{1.8cm}|c|c|c|c|c|c|c|c|c|} \hline Fronthaul capacity per sector~(Gbps) & 0.4 & 2 & 4 & 8 & 20 & 40 & 60 & 100 & 200 \\ \hline 
Number of streams $I_k$ & 1 & 2 & 2 & 2 & 2 & 2 & 3 & 3 & 4 \\ \hline \end{tabular}
\end{table}
\section{Conclusion}\label{Sec:mmWLSumm}
In this paper, we introduce a new architecture for mmWave CRAN with lens antenna arrays at the RRHs. We propose a low-complexity bit allocation scheme for SQ at the RRHs based only on the estimated received power levels of the antennas. We show that when the fronthaul is constrained, 
the proposed system which uses single-carrier transmission by the users combined with path delay compensation and reduced-size MMSE beamforming, can achieve large throughput gains over the conventional system that uses UPAs at the RRHs along with OFDM transmission by the users and full-size MMSE combining on each SC. Further, we propose a simple, yet effective channel estimation scheme for the proposed system, which exploits the unique energy focusing property of the lens array, and show that the 
proposed system still offers considerable advantage over the benchmark system, when both have imperfect CSI. 
Moreover, these gains are achieved at a much lower signal processing complexity at the CU, since it is sufficient to process a considerably fewer number of streams for the proposed system when compared to the benchmark. This shows that the proposed system with lens antenna arrays is a promising candidate for future evolution of CRAN operating in the mmWave frequencies. 

%

%
\appendix
As problem~\eqref{P:BARel} is convex and is strictly feasible, 
it can be optimally solved by solving its dual problem. Let $\lambda$ denote the Lagrange multiplier for constraint~\eqref{C:FHR} in problem~\eqref{P:BARel}. 
The Lagrangian is then given by
\begin{align}
L(\bm b_m,\lambda)=\sum_{q\in\mathcal{Q}}\frac{3\rho_{q,m}}{2^{2b_{q,m}}}+\lambda\Big(\sum_{q\in\mathcal{Q}}b_{q,m}-\frac{\bar{R}_m}{2W}\Big),\label{E:mmWLLagBA}
\end{align}
and the dual function is given by $f\left(\lambda\right)=\min_{\bm b_m\geq \bm 0}L(\bm b_m,\lambda)$. 
Differentiating~\eqref{E:mmWLLagBA} with respect to each $b_{q,m}$ and setting the derivative equal to zero, we have 
\begin{align}
b_{q,m}=\frac{1}{2}\log_2\Big(\frac{6\rho_{q,m}\ln 2}{\lambda}\Big).
\label{E:mmWLBALambda}
\end{align}
Since $b_{q,m}\geq 0\,\forall q\in\mathcal{Q}$, the dual variable $\lambda$ must lie in the interval $[0,6\rho_{\mathrm{max},m}\ln 2]$, 
where $\rho_{\mathrm{max},m}=\max_{q\in\mathcal{Q}}\rho_{q,m}$. Moreover, according to~\eqref{E:mmWLBAMinLag}, as $\lambda$ increases, the $b_{q,m}$'s decrease, making the constraint~\eqref{C:FHR} 
more feasible, and vice versa. Thus, the optimal 
$\lambda^\star$ that solves the dual problem $\max_{\lambda\geq 0} f(\lambda)$, can be found by a simple bisection search over $[0,6\rho_{\mathrm{max},m}\ln 2]$; 
then the optimal solution $\bm b'_m$ for 
problem~\eqref{P:BARel} is obtained 
as given by~\eqref{E:mmWLBAMinLag}. The proof of~\cref{Prop:mmWLOptBARel} is thus completed. 
%
%

\ifCLASSOPTIONcaptionsoff
  \newpage
\fi



\bibliographystyle{IEEEtran}
\bibliography{IEEEabrv,bibJournalList,ThesisBibliography}

\begin{thebibliography}{10}
\providecommand{\url}[1]{#1}
\csname url@samestyle\endcsname
\providecommand{\newblock}{\relax}
\providecommand{\bibinfo}[2]{#2}
\providecommand{\BIBentrySTDinterwordspacing}{\spaceskip=0pt\relax}
\providecommand{\BIBentryALTinterwordstretchfactor}{4}
\providecommand{\BIBentryALTinterwordspacing}{\spaceskip=\fontdimen2\font plus
\BIBentryALTinterwordstretchfactor\fontdimen3\font minus
  \fontdimen4\font\relax}
\providecommand{\BIBforeignlanguage}[2]{{%
\expandafter\ifx\csname l@#1\endcsname\relax
\typeout{** WARNING: IEEEtran.bst: No hyphenation pattern has been}%
\typeout{** loaded for the language `#1'. Using the pattern for}%
\typeout{** the default language instead.}%
\else
\language=\csname l@#1\endcsname
\fi
#2}}
\providecommand{\BIBdecl}{\relax}
\BIBdecl

\bibitem{andrews-etal2014what}
J.~Andrews \emph{et~al.}, ``What will {5G} be?'' \emph{{IEEE} J. Sel. Areas
  Commun.}, vol.~32, no.~6, pp. 1065--1082, June 2014.

\bibitem{zhou-yu2014optimized}
Y.~Zhou and W.~Yu, ``Optimized backhaul compression for uplink cloud radio
  access network,'' \emph{IEEE {J.} Sel.\ Areas Commun.}, vol.~32, no.~6, pp.
  1295--1307, June 2014.

\bibitem{liu-zhang2015optimized}
L.~Liu and R.~Zhang, ``Optimized uplink transmission in multi-antenna {C-RAN}
  with spatial compression and forward,'' \emph{{IEEE} Trans. Signal Process.},
  vol.~63, no.~19, pp. 5083--5095, Oct. 2015.

\bibitem{dai-yu2014sparse}
B.~Dai and W.~Yu, ``Sparse beamforming and user-centric clustering for downlink
  cloud radio access network,'' \emph{IEEE Access}, vol.~2, pp. 1326--1339,
  2014.

\bibitem{shi-etal2014group}
Y.~Shi, J.~Zhang, and K.~Letaief, ``Group sparse beamforming for green
  cloud-{RAN},'' \emph{{IEEE} Trans. Wireless Commun.}, vol.~13, no.~5, pp.
  2809--2823, May 2014.

\bibitem{luo-etal2015downlink}
S.~Luo, R.~Zhang, and T.~J. Lim, ``Downlink and uplink energy minimization
  through user association and beamforming in {C-RAN},'' \emph{{IEEE} Trans.
  Wireless Commun.}, vol.~14, no.~1, pp. 494--508, Jan. 2015.

\bibitem{stephen-zhang2017joint}
R.~G. Stephen and R.~Zhang, ``Joint millimeter-wave fronthaul and {OFDMA}
  resource allocation in ultra-dense {CRAN},'' \emph{{IEEE} Trans. Commun.},
  vol.~65, no.~3, pp. 1411--1423, Mar. 2017.

\bibitem{liu-etal2015joint}
L.~Liu, S.~Bi, and R.~Zhang, ``Joint power control and fronthaul rate
  allocation for throughput maximization in {OFDMA}-based cloud radio access
  network,'' \emph{{IEEE} Trans. Commun.}, vol.~63, no.~11, pp. 4097--4110,
  Nov. 2015.

\bibitem{stephen-zhang2017fronthaul}
R.~G. Stephen and R.~Zhang, ``Fronthaul-limited uplink {OFDMA} in ultra-dense
  {CRAN} with hybrid decoding,'' \emph{{IEEE} Trans. Veh. Technol.}, vol.~66,
  no.~10, pp. 9074--9084, Oct. 2017.

\bibitem{bai-etal2014coverage}
T.~Bai, A.~Alkhateeb, and R.~W. Heath, ``Coverage and capacity of
  millimeter-wave cellular networks,'' \emph{{IEEE} Commun. Mag.}, vol.~52,
  no.~9, pp. 70--77, Sep. 2014.

\bibitem{akdeniz-etal2014millimeter}
M.~R. Akdeniz \emph{et~al.}, ``Millimeter wave channel modeling and cellular
  capacity evaluation,'' \emph{{IEEE} J. Sel. Areas Commun.}, vol.~32, no.~6,
  pp. 1164--1179, June 2014.

\bibitem{hoydis-etal2011optimal}
J.~Hoydis, M.~Kobayashi, and M.~Debbah, ``Optimal channel training in uplink
  network {MIMO} systems,'' \emph{{IEEE} Trans. Signal Process.}, vol.~59,
  no.~6, pp. 2824--2833, Jun. 2011.

\bibitem{kang-etal2014joint}
J.~Kang, O.~Simeone, J.~Kang, and S.~S. Shitz, ``Joint signal and channel state
  information compression for the backhaul of uplink network mimo systems,''
  \emph{{IEEE} Trans. Wireless Commun.}, vol.~13, no.~3, pp. 1555--1567, Mar.
  2014.

\bibitem{zhang-etal2017locally}
J.~Zhang, X.~Yuan, and Y.~J. Zhang, ``Locally orthogonal training design for
  cloud-{RANs} based on graph coloring,'' \emph{{IEEE} Trans. Wireless
  Commun.}, vol.~16, no.~10, pp. 6426--6437, Oct. 2017.

\bibitem{gao-etal2016channel}
Z.~Gao, C.~Hu, L.~Dai, and Z.~Wang, ``Channel estimation for millimeter-wave
  massive {MIMO} with hybrid precoding over frequency-selective fading
  channels,'' \emph{IEEE Communications Letters}, vol.~20, no.~6, pp.
  1259--1262, Jun. 2016.

\bibitem{venugopal-etal2017channel}
K.~Venugopal, A.~Alkhateeb, N.~G. Prelcic, and R.~W. Heath, ``Channel
  estimation for hybrid architecture-based wideband millimeter wave systems,''
  \emph{{IEEE} J. Sel. Areas Commun.}, vol.~35, no.~9, pp. 1996--2009, Sep.
  2017.

\bibitem{zeng-etal2014electromagnetic}
Y.~Zeng, R.~Zhang, and Z.~N. Chen, ``Electromagnetic lens-focusing antenna
  enabled massive {MIMO}: Performance improvement and cost reduction,''
  \emph{{IEEE} J. Sel. Areas Commun.}, vol.~32, no.~6, pp. 1194--1206, June
  2014.

\bibitem{zeng-zhang2016millimeter}
Y.~Zeng and R.~Zhang, ``Millimeter wave {MIMO} with lens antenna array: A new
  path division multiplexing paradigm,'' \emph{{IEEE} Trans. Commun.}, vol.~64,
  no.~4, pp. 1557--1571, Apr. 2016.

\bibitem{zeng-zhang2017cost}
------, ``Cost-effective millimeter-wave communications with lens antenna
  array,'' \emph{IEEE Wireless Communications}, vol.~24, no.~4, pp. 81--87,
  2017.

\bibitem{zeng-etal2016multiuser}
Y.~Zeng, L.~Yang, and R.~Zhang, ``Multi-user millimeter wave {MIMO} with
  full-dimensional lens antenna array,'' \emph{{IEEE} Trans. Wireless Commun.},
  vol.~17, no.~4, pp. 2800--2814, Apr. 2018.

\bibitem{kwon-etal2016rf}
T.~Kwon, Y.~G. Lim, B.~W. Min, and C.~B. Chae, ``{RF} lens-embedded massive
  {MIMO} systems: Fabrication issues and codebook design,'' \emph{{IEEE} Trans.
  Microw. Theory Tech.}, vol.~64, no.~7, pp. 2256--2271, July 2016.

\bibitem{yang-etal2017channel}
L.~Yang, Y.~Zeng, and R.~Zhang, ``Channel estimation for millimeter-wave {MIMO}
  communications with lens antenna arrays,'' \emph{{IEEE} Trans. Veh.
  Technol.}, vol.~67, no.~4, pp. 3239--3251, Apr. 2018.

\bibitem{brady-etal2013beamspace}
J.~Brady, N.~Behdad, and A.~M. Sayeed, ``Beamspace {MIMO} for millimeter-wave
  communications: System architecture, modeling, analysis, and measurements,''
  \emph{{IEEE} Trans. Antennas Propag.}, vol.~61, no.~7, pp. 3814--3827, Jul.
  2013.

\bibitem{gao-etal2017reliable}
X.~Gao, L.~Dai, S.~Han, C.~L. I, and X.~Wang, ``Reliable beamspace channel
  estimation for millimeter-wave massive {MIMO} systems with lens antenna
  array,'' \emph{IEEE Transactions on Wireless Communications}, vol.~16, no.~9,
  pp. 6010--6021, Sept 2017.

\bibitem{narasimhamurthy-tepedelenlioglu2009antenna}
A.~B. Narasimhamurthy and C.~Tepedelenlioglu, ``Antenna selection for
  {MIMO-OFDM} systems with channel estimation error,'' \emph{{IEEE} Trans. Veh.
  Technol.}, vol.~58, no.~5, pp. 2269--2278, Jun. 2009.

\bibitem{tse2005fundamentals}
D.~Tse and P.~Viswanath, \emph{Fundamentals of wireless communication}.\hskip
  1em plus 0.5em minus 0.4em\relax Cambridge University Press, 2005.

\bibitem{kay1993fundamentals}
S.~M. Kay, \emph{Fundamentals of Statistical Signal Processing: Estimation
  Theory}.\hskip 1em plus 0.5em minus 0.4em\relax Upper Saddle River, NJ, USA:
  Prentice-Hall, Inc., 1993.

\bibitem{popovic1992generalized}
B.~M. Popovic, ``Generalized chirp-like polyphase sequences with optimum
  correlation properties,'' \emph{{IEEE} Trans. Inf. Theory}, vol.~38, no.~4,
  pp. 1406--1409, Jul. 1992.

\bibitem{sesia2009lte}
S.~Sesia, I.~Toufik, and M.~Baker, \emph{LTE, The UMTS Long Term Evolution:
  From Theory to Practice}.\hskip 1em plus 0.5em minus 0.4em\relax Wiley
  Publishing, 2009.

\bibitem{chi-etal2011training}
Y.~Chi, A.~Gomaa, N.~Al-Dhahir, and A.~R. Calderbank, ``Training signal design
  and tradeoffs for spectrally-efficient multi-user {MIMO-OFDM} systems,''
  \emph{{IEEE} Trans. Wireless Commun.}, vol.~10, no.~7, pp. 2234--2245, Jul.
  2011.

\bibitem{samimi-rappaport2014ultra}
M.~K. Samimi and T.~S. Rappaport, ``Ultra-wideband statistical channel model
  for non line of sight millimeter-wave urban channels,'' in \emph{Proc.\ {IEEE
  GLOBECOM}}, Dec. 2014, pp. 3483--3489.

\end{thebibliography}
\end{document}